\documentclass[a4paper,USenglish,cleveref,numberwithinsect,autoref,thm-restate]{lipics-v2021}
\usepackage[utf8]{inputenc}
\usepackage{color}
\usepackage{csquotes}
\usepackage{xspace}
\usepackage{amsmath}
\usepackage{amssymb}
\usepackage{mathtools}
\usepackage{bbm}
\usepackage[linesnumbered, noend]{algorithm2e}
\usepackage[draft]{fixme}
\usepackage{tikz}
\usetikzlibrary{arrows,arrows.meta,decorations.pathmorphing}
\usepackage{bbding}
\usepackage{xcolor,pifont}
\usepackage[noadjust]{cite}

\nolinenumbers

\hideLIPIcs

\fxsetup{mode=multiuser,theme=color, layout=inline}
\FXRegisterAuthor{eg}{eg}{\color{red} {\bf  Egor}}
\FXRegisterAuthor{tk}{atk}{\color{blue} {\bf Tomasz}}

\definecolor{darkgreen}{RGB}{0,160,0}
\definecolor{darkred}{RGB}{220,20,60}
\definecolor{darkblue}{RGB}{0,0,160}

\newcommand{\Oh}{\ensuremath{\mathcal{O}}\xspace}
\newcommand{\Ohtilde}{\ensuremath{\smash{\rlap{\raisebox{-0.2ex}{$\widetilde{\phantom{\Oh}}$}}\Oh}}\xspace}
\newcommand{\tOh}{\Ohtilde}

\newcommand{\on}[1]{\operatorname{#1}}

\DeclareMathOperator*{\argmin}{arg\,min}

\DeclareMathOperator{\poly}{poly}

\DeclareMathOperator{\dist}{dist}
\DeclareMathOperator{\polylog}{polylog}

\newcommand{\mds}{\mathsf{CMO}}

\newcommand{\dd}{.\,.}
\newcommand{\core}{\on{core}}

\newcommand{\dens}[1]{#1^{\square}}

\newcommand{\lo}{\textup{lo}}
\newcommand{\hi}{\textup{hi}}
\newcommand{\minplus}{\otimes}
\newcommand{\deltasigma}{\delta^{\Sigma}}
\newcommand{\wit}[2]{\mathcal{W}^{#1, #2}}
\newcommand{\corei}[1]{\core_{#1, \cdot}}
\newcommand{\corej}[1]{\core_{\cdot, #1}}
\newcommand{\deltai}[1]{\delta_{#1, \cdot}}
\newcommand{\deltaj}[1]{\delta_{\cdot, #1}}
\newcommand{\lco}{\mathsf{lco}}

\def\twoheadleadsto{\tikz[baseline=(a.base)]{\draw[%
    decorate,decoration={zigzag,segment length=4, amplitude=.9},%
    ] (0,0) -- (.25, 0);%
    \draw[%
    -{Classical TikZ Rightarrow}.{Classical TikZ Rightarrow},%
    ] (.25, 0) -- (.4, 0);%
    \node (a) at (.4/2,-.03) {\phantom{\(\leadsto\)}};%
}}

\newcommand{\onto}{\twoheadleadsto}
\def\aonto#1{\onto}

\newcommand{\fragmentco}[2]{[#1\dd #2)}
\newcommand{\fragmentoc}[2]{(#1\dd #2]}
\newcommand{\fragmentoo}[2]{(#1\dd #2)}
\newcommand{\fragmentcc}[2]{[#1\dd #2]}
\newcommand{\position}[1]{[#1]}

\newcommand{\floor}[1]{\lfloor{#1}\rfloor}
\newcommand{\ceil}[1]{\lceil{#1}\rceil}

\newcommand{\lis}{\mathsf{LIS}}
\newcommand{\plh}{\ensuremath{\bigstar}}
\newcommand{\pos}{\on{pos}}
\newcommand{\rpos}{\on{rpos}}
\newcommand{\maybe}{}

\crefname{observation}{Observation}{Observations}
\newtheorem{fact}[theorem]{Fact}
\crefname{fact}{Fact}{Facts}

\title{Core-Sparse Monge Matrix Multiplication}
\subtitle{Improved Algorithm and Applications}
\titlerunning{Core-Sparse Monge Matrix Multiplication: Improved Algorithm and Applications}

\author{Pawe\l{} Gawrychowski}{University of Wroc\l{}aw, Poland}{gawry@cs.uni.wroc.pl}{https://orcid.org/0000-0002-6993-5440}{}%{[funding]}

\author{Egor Gorbachev}{Saarland University and Max Planck Institute for Informatics, Saarland Informatics Campus, Germany}{egorbachev@cs.uni-saarland.de}{https://orcid.org/0009-0005-5977-7986}{This work is part of the project TIPEA that has received funding from the European Research Council (ERC) under the European Unions Horizon 2020 research and innovation programme (grant agreement No.\ 850979).}%TODO mandatory, please use full name; only 1 author per \author macro; first two parameters are mandatory, other parameters can be empty. Please provide at least the name of the affiliation and the country. The full address is optional. Use additional curly braces to indicate the correct name splitting when the last name consists of multiple name parts.

\author{Tomasz Kociumaka}{Max Planck Institute for Informatics, Saarland Informatics Campus, Germany}{tomasz.kociumaka@mpi-inf.mpg.de}{https://orcid.org/0000-0002-2477-1702}{}

\authorrunning{P. Gawrychowski, E. Gorbachev, T. Kociumaka} %TODO mandatory. First: Use abbreviated first/middle names. Second (only in severe cases): Use first author plus 'et al.'

\Copyright{Pawe\l{} Gawrychowski, Egor Gorbachev, Tomasz Kociumaka} %TODO mandatory, please use full first names. LIPIcs license is "CC-BY";  http://creativecommons.org/licenses/by/3.0/

\acknowledgements{}

\begin{document}

\maketitle

\begin{abstract}
    \emph{Min-plus matrix multiplication} is a fundamental tool for designing algorithms operating on
distances in graphs and different problems solvable by dynamic programming.
We know that, assuming the APSP hypothesis, no subcubic-time algorithm exists for the case of general matrices. 
However, in many applications the matrices admit certain structural properties that can be used to design faster algorithms.
For example, when considering a planar graph, one often works with a \emph{Monge matrix}~$A$, meaning
that the \emph{density matrix} $\dens{A}$ has non-negative entries, that is, $\dens{A}_{i,j}\coloneqq A_{i+1,j}+A_{i,j+1}-A_{i,j}-A_{i+1,j+1} \ge 0$.
The min-plus product of two $n\times n$ Monge matrices can be computed in $\Oh(n^2)$ time using the famous SMAWK algorithm.

In applications such as longest common subsequence, edit distance, and longest increasing subsequence, the matrices are even more structured, as observed by Tiskin [J. Discrete Algorithms, 2008]: 
they are (or can be converted to) \emph{simple unit-Monge matrices}, meaning that the density matrix is a permutation matrix and, furthermore, the first column and the last row of the matrix consist of only zeroes.
Such matrices admit an implicit representation of size $\Oh(n)$ and, as shown by Tiskin [SODA 2010 \& Algorithmica, 2015], their
min-plus product can be computed in $\Oh(n\log n)$ time.
Russo [SPIRE 2010 \& Theor. Comput. Sci., 2012] identified a general structural property of matrices that admit such efficient
representation and min-plus multiplication algorithms: the \emph{core size} $\delta$, defined as the number of non-zero entries
in the density matrices of the input and output matrices. 
He provided an adaptive implementation of the SMAWK algorithm that runs in $\Oh((n+\delta)\log^3 n)$ or $\Oh((n+\delta)\log^2 n)$ time
(depending on the representation of the input matrices).

In this work, we further investigate the core size as the parameter that enables efficient min-plus matrix multiplication. 
On the combinatorial side, we provide a (linear) bound on the core size of the product matrix in terms of the core sizes of the input matrices. 
On the algorithmic side, we generalize Tiskin's algorithm (but, arguably, with a more elementary analysis) to
solve the core-sparse Monge matrix multiplication problem in $\Oh(n+\delta\log \delta) \subseteq \Oh(n + \delta \log n)$ time, matching the complexity
for simple unit-Monge matrices.
As witnessed by the recent work of Gorbachev and Kociumaka [STOC'25] for edit distance with integer weights, our generalization opens up the possibility of speed-ups for weighted sequence alignment problems.
Furthermore, our multiplication algorithm is also capable of producing an efficient data structure for recovering the witness for any given entry of the output matrix.
This allows us, for example, to preprocess an integer array of size $n$ in $\Ohtilde(n)$ time so that the longest increasing subsequence
of any sub-array can be reconstructed in $\Ohtilde(\ell)$ time, where $\ell$ is the length of the reported subsequence.
In comparison, Karthik C.\ S.\ and Rahul [arXiv, 2024] recently achieved $\Oh(\ell+n^{1/2}\polylog n)$-time reporting after
$\Oh(n^{3/2}\polylog n)$-time preprocessing.

\end{abstract}

\section{Introduction}\label{sec:introduction}
The \emph{min-plus product} (also known as the distance product or the tropical product) of two matrices $A$ and $B$ is defined as a matrix $C=A\minplus B$ such that $C_{i,k} = \min_{j} \left(A_{i,j}+B_{j,k}\right)$. 
The task of computing the min-plus product of two $n\times n$ matrices can be solved in \smash{$n^{3}/\exp(\Omega(\sqrt{\log n}))$} time~\cite{Wil18}, and it is fine-grained equivalent to the All-Pairs Shortest Path (APSP) problem~\cite{VW18}, asking to compute the distances between every pair of vertices in a directed weighted graph on $n$ vertices.
While it is conjectured that APSP, and hence also the min-plus product, do not admit $n^{3-\Omega(1)}$-time solutions, faster algorithms exist for many special cases arising in numerous applications of the min-plus product; see, e.g.,~\cite{AGM97,VW06,Yus09,Cha10,BGSW19,VWX20,GPVWZ21,CDX22,CDXZ22,Due23}.
Most of the underlying procedures rely on fast matrix multiplication for the standard $(+,\cdot)$-product.

\emph{Monge matrices} constitute a notable exception:
An $n\times n$ matrix $A$ is a Monge matrix if its \emph{density matrix} $\dens{A}$ is non-negative, that is, $\dens{A}_{i,j}\coloneqq A_{i+1,j}+A_{i,j+1}-A_{i,j}-A_{i+1,j+1}\ge 0$ holds for all $i,j\in [0\dd n-1)$.
The min-plus product of two $n\times n$ Monge matrices can be computed in $\Oh(n^2)$ time using the SMAWK algorithm~\cite{AKMSW87}, and the resulting matrix still satisfies the Monge property. 
Monge matrices arise in many combinatorial optimization problems; see \cite{BKR96,Bur07} for surveys.
One of the most successful applications is for planar graphs, where the distances between vertices on a single face satisfy the Monge property (see \cref{fct:monge-planar} and \cref{fig:monge-planar}).
This observation allowed for an $\Ohtilde(n)$-time\footnote{Throughout this paper, we use $\Ohtilde(\cdot)$ notation to suppress factors poly-logarithmic in the input size.} single-source shortest path algorithm for planar graphs (with negative real weights)~\cite{FR06}, and the resulting techniques now belong to the standard toolkit for designing planar graph algorithms; see, e.g.,~\cite{BKMNW17,IKLS17,CGLMPWW23}.

Another important application of Monge matrices is in sequence alignment problems, such as \emph{edit distance} and \emph{longest common subsequence} (LCS), as well as in the related \emph{longest increasing subsequence} (LIS) problem.
Already in the late 1980s, Apostolico, Atallah, Larmore, and McFaddin~\cite{AALM90} noted that the so-called \emph{DIST matrices}, which (among others) specify the weighted edit distances between prefixes of one string and suffixes of another string, satisfy the Monge property. 
A modern interpretation of this observation is that these matrices store boundary-to-boundary distances in planar \emph{alignment graphs}~\cite{Schm98}; see also~\cite{Ben95,KM96,LZ01} for further early applications of DIST matrices and their Monge property.

In the late 2000s, Tiskin~\cite{Tis08b,Tis08a} observed that the DIST matrices originating from the unweighted variants of edit distance, LCS, and LIS problems are more structured.
For this, he introduced the notions of \emph{unit-Monge matrices}, whose density matrices are permutation matrices (that is, binary matrices whose rows and columns contain exactly a single $1$ entry each) and \emph{simple Monge matrices}, whose leftmost column and bottommost row consist of zeroes.
He also proved that the product of two simple unit-Monge matrices still belongs to this class and can be computed in $\Oh(n\log n)$ time provided that each matrix $A$ is represented using the underlying permutation~$P_A$~\cite{Tiskin10}.
By now, the resulting algorithm has found numerous applications, including for computing LCS and edit distance of compressed strings~\cite{Gaw12,HLLW13,Tiskin10,GKLS22}, maintaining these similarity measures for dynamic strings~\cite{CKM20,GK24}, approximate pattern matching~\cite{Tis14,CKW22}, parallel and distributed algorithms for similarity measures~\cite{KT10,MBT21}, and oracles for substring similarity~\cite{Sak19,CGMW21,Sak22}.
Furthermore, Tiskin's algorithm has been used to solve the LIS problem in various settings, such as dynamic~\cite{KS21}, parallel~\cite{CHS23}, and distributed~\cite{Koo24}.
A disadvantage of Tiskin's original description (and even the later informal descriptions~\cite{CF}) is its dependence on the algebraic structure known as the monoid of seaweed braids, which natively supports unweighted LCS only (tasks involving edit distance need to be reduced to LCS counterparts).
This makes the algorithm difficult to generalize to weighted problems and extend even to seemingly
simple questions such as recovering (an implicit representation of) the witness indices $j$ such that $C_{i,k} = A_{i,j}+B_{j,k}$~\cite{KR24}.

Russo~\cite{Russo10} identified the number of non-zero elements in $\dens{A}$, $\dens{B}$, and $\dens{(A\minplus B)}$,
called the \emph{core size} $\delta$, as the parameter that enables fast min-plus matrix multiplication. It is easy to see
that $A$, $B$, and $A\minplus B$ can be stored in $\Oh(n+\delta)$ space using a \emph{condensed representation}, consisting of the
\emph{boundary entries} (the leftmost column and bottommost row, e.g., suffice) and \emph{core elements}, i.e., the non-zero entries of the density matrix.
Then, Russo's algorithm for the \emph{core-sparse Monge matrix multiplication} problem computes the condensed representation of
$A\minplus B$ in $\Oh((n+\delta)\log^2 n)$ time when provided with constant-time random access to $A$ and $B$, and in
$\Oh((n+\delta)\log^3 n)$ time given condensed representations of $A$ and $B$.\footnote{More precisely, Russo's algorithm builds a data structure that provides $\Oh(\log n)$-time random access to the entries of $A \minplus B$. Consequently, for repeated multiplication, we cannot assume constant-time access to the input matrices. Hence, $\Oh((n+\delta)\log^3 n)$ is a more realistic bound.}
Russo's algorithm has a very different structure
than Tiskin's: it relies on a clever adaptation of the SMAWK algorithm~\cite{AKMSW87} to provide efficient random access to $C$ and then employs
binary search to find individual non-zero entries of $\dens{C}$.
This brings the question of unifying both approaches and understanding the complexity of core-sparse Monge matrix multiplication.

\subparagraph{Our Results.}
We consider the core-sparse Monge matrix multiplication problem from the combinatorial and algorithmic points of view, and confirm
that the core size is the right parameter that enables fast min-plus matrix multiplication.
Let $\delta(A)$, or the core size of~$A$, denote the number of non-zero elements in $\dens{A}$.
We begin with analyzing, in \cref{sec:core-properties}, how the core size of $A\minplus B$ depends on the core sizes of $A$ and $B$.

\begin{restatable}{theorem}{lmcorepreservation} \label{lm:core-preservation}
    Let $A$ be a $p \times q$ Monge matrix, and let $B$ be a $q \times r$ Monge matrix.
    We have $\delta(A \minplus B) \le 2 \cdot (\delta(A) + \delta(B))$.
\end{restatable}
We stress that this the first bound on $\delta(A\minplus B)$ in terms of $\delta(A)$ and $\delta(B)$: the complexity analysis of Russo's algorithm~\cite{Russo10} requires a bound on all $\delta(A)$, $\delta(B)$, and $\delta(A\minplus B)$.

Next, in \cref{sec:core-based-multiplication-algorithm}, we generalize Tiskin's algorithm (but fully avoiding the formalism of the seaweed product) to solve the core-sparse Monge matrix multiplication problem.
We believe that the more elementary interpretation makes our viewpoint not only more robust but also easier to understand.
At the same time, the extension from simple unit-Monge matrices to general core-sparse Monge matrices introduces a few technical complications handled in~\cref{sec:appendix}.
Notably, we need to keep track of the leftmost column and the bottommost row instead of assuming they are filled with zeroes. Further, the core does not form a permutation so splitting it into two halves of the same size requires some calculations.

\newcommand{\matone}{A}
\newcommand{\mattwo}{B}
\newcommand{\dimone}{p}
\newcommand{\dimtwo}{q}
\newcommand{\dimthree}{r}
\newcommand{\myarticle}{a}
\newcommand{\tmpfootnote}{\footnote{In the technical sections of the paper we use the $\Oh(\cdot)$-notation conservatively. Specifically, we interpret $\Oh(f(x_1, \ldots, x_k))$ as the set of functions $g(x_1, \ldots, x_k)$ for which there are constants $c_g, N_g > 0$ such that $g(x_1, \ldots, x_k) \le c_g \cdot f(x_1, \ldots, x_k)$ holds for all valid tuples $(x_1, \ldots, x_k)$ satisfying $\max_i x_i \ge N_g$. Accordingly, whenever the expression inside $\Oh(\cdot)$ depends on multiple parameters, we sometimes add $1$ or $2$ to the arguments of logarithms to ensure formal correctness in corner cases.}}
\begin{restatable}{theorem}{lmcorebasedmultiplicationalgorithm} \label{lm:core-based-multiplication-algorithm}
    There is a (deterministic) algorithm that, given the condensed representations of \myarticle{} $\dimone \times \dimtwo$ Monge matrix $\matone$ and \myarticle{} $\dimtwo \times \dimthree$ Monge matrix $\mattwo$, in time $\Oh(\dimone + \dimtwo + \dimthree + (\delta(\matone) + \delta(\mattwo)) \log (1+\delta(\matone) + \delta(\mattwo)))$ computes the condensed representation of $\matone \minplus \mattwo$.\tmpfootnote
\end{restatable}
\renewcommand{\matone}{M_1}
\renewcommand{\mattwo}{M_2}
\renewcommand{\dimone}{n_1}
\renewcommand{\dimtwo}{n_2}
\renewcommand{\dimthree}{n_3}
\renewcommand{\myarticle}{an}
\renewcommand{\tmpfootnote}{}
The above complexity improves upon Russo's~\cite{Russo10} and matches Tiskin's~\cite{Tiskin10} for the simple unit-Monge case.
Thanks to the more direct description, we can easily extend our algorithm to build (in the same complexity)
an $\Oh(n+\delta)$-size data structure that, given $(i,k)$, in $\Oh(\log n)$ time computes the smallest \emph{witness} $j$ such that $C_{i,k}=A_{i,j}+B_{j,k}$.

\subparagraph{Applications.}
As an application of our witness recovery functionality, we consider the problem of range LIS queries.
This task asks to preprocess an integer array $s[0\dd n)$ so that, later, given two indices $0\le i < j \le n$, the longest increasing subsequence (LIS) of the sub-array $s[i\dd j)$ can be reported efficiently.
Tiskin~\cite{Tis07,Tiskin10} showed that the LIS size $\ell$ can be reported in $\Oh(\log n)$ time after $\Oh(n\log^2 n)$-time preprocessing. 
It was unclear, though, how to efficiently report the LIS itself. 
The recent work of Karthik C.\ S.\ and Rahul~\cite{KR24} achieved $\Oh(n^{1/2}\log^3 n+\ell)$-time reporting (correct with high probability) after $\Oh(n^{3/2}\log^3 n)$-time preprocessing. 
It is fairly easy to use our witness recovery algorithm to deterministically support $\Oh(\ell\log^2 n)$-time reporting after $\Oh(n\log^2 n)$-time preprocessing; see \cref{sec:range-lis} for an overview of this result.
As further shown in \cref{app:range-lis}, the reporting time can be improved to $\Oh(\ell \log n)$ and, with the preprocessing time increased to $\Oh(n\log^3 n)$, all the way to $\Oh(\ell)$.

\begin{restatable}{theorem}{thmlis}\label{lm:lis-reporting}
    For every parameter $\alpha \in [0, 1]$, there exists an algorithm that, given an integer array $s[0\dd n)$, in time $\Oh(n \log^{3 - \alpha} n)$ builds a data structure that can answer range LIS \emph{reporting} queries in time $\Oh(\ell \log^{\alpha} n)$, where $\ell$ is the length of the reported sequence.

    In particular, there is an algorithm with $\Oh(n \log^3 n)$ preprocessing and $\Oh(\ell)$ reporting time and an algorithm with $\Oh(n \log^2 n)$ preprocessing and $\Oh(\ell \log n)$ reporting time.
\end{restatable}

In parallel to this work, Gorbachev and Kociumaka~\cite{GK24} used core-sparse Monge matrix multiplication
for the \emph{weighted edit distance} with integer weights.
\Cref{lm:core-based-multiplication-algorithm} allowed for saving two logarithmic factors in the final time complexities compared to the initial preprint using Russo's approach~\cite{Russo10}.
Weighted edit distance is known to be reducible to unweighted LCS only for a very limited class of so-called uniform weight functions~\cite{Tis07}, so this application requires the general core-sparse Monge matrix multiplication.

\subparagraph{Open Problem.}

An interesting open problem is whether any non-trivial trade-off can be achieved for the \emph{weighted} version of range LIS queries, where each element of $s$ has a weight, and the task is to compute a maximum-weight increasing subsequence of $s[i\dd j)$: either the weight alone or the whole subsequence. 
Surprisingly, as we show in \cref{sec:open}, if our bound $\delta(A\minplus B)\le 2\cdot (\delta(A)+\delta(B))$ of \cref{lm:core-preservation} can be improved to $\delta(A\minplus B)\le c\cdot (\delta(A)+\delta(B))+\Ohtilde(p+q+r)$ for some $1\le c < 2$, then our techniques automatically yield a solution with $\Ohtilde(n^{1+\log_2 c})$ preprocessing time, $\Ohtilde(1)$ query time, and $\Ohtilde(\ell)$ reporting time.

\section{Preliminaries}\label{sec:preliminaries}

\begin{definition} \label{def:monge}
    A matrix $A$ of size $p \times q$ is a \emph{Monge matrix} if it satisfies the following \emph{Monge property}
    for every $i \in [0 \dd p - 1)$ and $j \in [0 \dd q - 1)$:
    \[A_{i, j} + A_{i + 1, j + 1} \le A_{i, j + 1} + A_{i + 1, j}.\]
    Furthermore, $A$ is an \emph{anti-Monge matrix} if the matrix $-A$ (with negated entries) is a Monge matrix.
    %\lipicsEnd
\end{definition}

In some sources, the following equivalent (but seemingly stronger) condition is taken as the definition of Monge matrices.
\begin{observation}\label{cor:non-local-monge-property}
    A matrix $A$ of size $p \times q$ is a Monge matrix if and only if it satisfies the following inequality for all integers $0\le a \le b < p$ and $0\le c \le d < q$:
    \[A_{a, c} + A_{b, d} \le A_{a, d} + A_{b, c}.\]
\end{observation}
\begin{proof}
    It suffices to sum the inequality in \cref{def:monge} for all $i\in [a\dd b)$ and $j\in [c\dd d)$.
\end{proof}

\begin{definition} \label{def:min-plus}
    The \emph{min-plus product} of a matrix $A$ of size $p \times q$ and a matrix $B$ of size $q \times r$ is a matrix $A \minplus B=C$ of size $p \times r$ satisfying $C_{i, k} = \min_{j\in \fragmentco{0}{q}} A_{i, j} + B_{j, k}$ for all $i \in \fragmentco{0}{p}$ and $k \in \fragmentco{0}{r}$.

    For $i \in \fragmentco{0}{p}$ and $k \in \fragmentco{0}{r}$, we call $j \in \fragmentco{0}{q}$ a \emph{witness} of $C_{i, k}$ if and only if $C_{i, k} = A_{i, j} + B_{j, k}$.
    We define the $p \times r$ \emph{witness matrix} $\wit{A}{B}$ such that $\wit{A}{B}_{i, k}$ is the smallest witness of $C_{i, k}$ for each $i \in \fragmentco{0}{p}$ and $k \in \fragmentco{0}{r}$.
    %\lipicsEnd
\end{definition}

The min-plus product of two Monge matrices is also Monge; we include a simple proof in \cref{sec:appendix-fctmongeproductismonge} for completeness.

\renewcommand{\maybe}{\lipicsEnd}
\begin{restatable}[{\cite[Corollary A]{Yao82}}]{fact}{fctmongeproductismonge} \label{fct:monge-product-is-monge}
    Let $A, B$, and $C$ be matrices such that $A \minplus B = C$.
    If $A$ and $B$ are Monge, then $C$ is also Monge.
    %\maybe
\end{restatable}
\renewcommand{\maybe}{}

In the context of Monge matrices, it is useful to define the core and the density matrix.

\begin{definition}
    The \emph{density matrix} of a matrix $A$ of size $p \times q$ is a matrix $\dens{A}$ of size $(p - 1) \times (q - 1)$ satisfying $\dens{A}_{i, j} = A_{i, j + 1} + A_{i + 1, j} - A_{i, j} - A_{i + 1, j + 1}$ for $i \in \fragmentco{0}{p - 1}, j \in \fragmentco{0}{q - 1}$.

    We define the \emph{core} of the matrix $A$ as \[\core(A) \coloneqq \{(i, j, \dens{A}_{i, j}) \mid i \in \fragmentco{0}{p - 1}, j \in \fragmentco{0}{q - 1}, \dens{A}_{i, j} \neq 0 \}\] and denote the \emph{core size} of the matrix $A$ by $\delta(A) \coloneqq |\core(A)|$.
    Furthermore, we define the \emph{core sum} of the matrix $A$ as the sum of the values of all core elements of $A$, that is, \[\deltasigma(A) \coloneqq \sum_{i \in \fragmentco{0}{p - 1}} \sum_{j \in \fragmentco{0}{q - 1}} \dens{A}_{i, j}.\]
    %\lipicsEnd
\end{definition}

Note that, for a Monge matrix $A$, all entries of its density matrix are non-negative, and thus $\core(A)$ consists of triples $(i, j, v)$ with some positive values $v$.

For any matrix $A$ of size $p \times q$ and integers $0\le a < b \le p$ and $0 \le c < d \le q$, we write $A\fragmentco{a}{b}\fragmentco{c}{d}$ to denote the contiguous submatrix of $A$ consisting of all entries on the intersection of rows $\fragmentco{a}{b}$ and columns $\fragmentco{c}{d}$ of $A$.
Matrices $A\fragmentcc{a}{b}\fragmentcc{c}{d}$, $A\fragmentoc{a}{b}\fragmentco{c}{d}$, etc., are defined analogously.

\section{\texorpdfstring{\boldmath Properties of $\delta$ and $\deltasigma$}{Properties of δ and δ^Σ}}\label{sec:core-properties}

In this section, we provide some useful properties of $\delta$ and $\deltasigma$.
Most importantly, we show how to bound $\delta(A \minplus B)$ in terms of $\delta(A)$ and $\delta(B)$.

\noindent
The following observation is a straightforward consequence of the definitions of $\deltasigma$ and $\dens{A}$.

\renewcommand{\maybe}{\lipicsEnd}
\begin{restatable}{observation}{lmcalculatevaluefromcore}\label{lm:calculate-value-from-core}
    For a matrix $A$ of size $p \times q$, and integers $0\le a \le b < p$ and $0\le c \le d < q$,
    \[A_{a, c} + A_{b, d} + \deltasigma(A\fragmentcc{a}{b}\fragmentcc{c}{d}) = A_{a, d} + A_{b, c}. \]%\maybe\]
\end{restatable}
\renewcommand{\maybe}{}

\begin{proof}
    By the definition of $\dens{A}$, we have $A_{i, j} + A_{i + 1, j + 1} + \dens{A}_{i, j} = A_{i, j + 1} + A_{i + 1, j}$ for all $i \in \fragmentco{a}{b}$ and $j \in \fragmentco{c}{d}$.
    The desired equality follows by summing up all these equalities.
\end{proof}

An application of \cref{lm:calculate-value-from-core} for $a = c = 0$ implies that every value of $A$ can be uniquely reconstructed from $A$'s core and the values in the topmost row and the leftmost column of $A$.

\begin{definition}
    The \emph{condensed representation} of a matrix $A$ consists of the core of $A$ as well as the values in the topmost row and the leftmost column of $A$. 
    %\lipicsEnd
\end{definition}

\begin{observation}\label{lm:monge-submatrix}
    Any submatrix $A'$ (not necessarily contiguous) of a Monge matrix $A$ is Monge.
\end{observation}
\begin{proof}
    By \cref{cor:non-local-monge-property}, if $A$ is Monge, the Monge property holds for every (not necessarily contiguous) $2 \times 2$ submatrix of $A$. 
    In particular, it holds for contiguous $2 \times 2$ submatrices of $A'$, and thus $A'$ satisfies \cref{def:monge}.
\end{proof}

We next show a crucial property of the witness matrix.

\begin{restatable}[{\cite[Theorem 1]{Yao82}}]{lemma}{lmwitnessmonotonicity} \label{lm:witness-monotonicity}
    For any two Monge matrices, $A$ of size $p \times q$ and $B$ of size $q \times r$, the witness matrix $\wit{A}{B}$ is non-decreasing by rows and columns.
\end{restatable}

\begin{proof}
    We prove that $\wit{A}{B}$ is non-decreasing by columns.
    The claim for rows is analogous.
    Fix some $i \in \fragmentco{0}{p-1}$ and $k \in \fragmentco{0}{r}$.
    Let $j^* \coloneqq \wit{A}{B}_{i, k}$ and $j\in \fragmentco{0}{j^*}$.
    As $j^*$ is the smallest witness for $(i, k)$, we have $A_{i, j} + B_{j, k} > A_{i, j^*} + B_{j^*, k}$.
    Using this fact and the Monge property for $A$, we derive
    \begin{align*}
        A_{i + 1, j} + B_{j, k} &\ge (A_{i, j} + A_{i + 1, j^*} - A_{i, j^*}) + B_{j, k}\\
                                &= (A_{i, j} + B_{j, k}) + (A_{i + 1, j^*} - A_{i, j^*})\\
                                &> (A_{i, j^*} + B_{j^*, k}) + (A_{i + 1, j^*} - A_{i, j^*})\\
                                &= A_{i + 1, j^*} + B_{j^*, k}.
    \end{align*}
    Since $A_{i + 1, j} + B_{j, k}>A_{i + 1, j^*} + B_{j^*, k}$ for all $j \in \fragmentco{0}{j^*}$, we conclude that $\wit{A}{B}_{i + 1, k} \ge j^*=\wit{A}{B}_{i, k}$ holds as required.
\end{proof}

We now show how to bound $\deltasigma(A \minplus B)$ in terms of $\deltasigma(A)$ and $\deltasigma(B)$.

\begin{lemma}\label{lm:core-sum-preservation}
    Let $A$ be a $p \times q$ Monge matrix, and let $B$ be a $q \times r$ Monge matrix.
    Then, $\deltasigma(A \minplus B) \le \min\{\deltasigma(A), \deltasigma(B)\}$.
\end{lemma}

\begin{proof}
    Let $C \coloneqq A \minplus B$.
    Denote $j \coloneqq \wit{A}{B}_{0, 0}$ and $j' \coloneqq \wit{A}{B}_{p - 1, r - 1}$, where $j \le j'$ due to \cref{lm:witness-monotonicity}.
    We have $C_{0, 0} = A_{0, j} + B_{j, 0}$ and $C_{p - 1, r - 1} = A_{p - 1, j'} + B_{j', r - 1}$.
    Furthermore, $C_{p - 1, 0} \le A_{p - 1, j} + B_{j, 0}$ and $C_{0, r - 1} \le A_{0, j'} + B_{j', r - 1}$ due to the definition of $C$.
    Hence, due to \cref{lm:calculate-value-from-core}, we obtain
    \begin{align*}
        \deltasigma(C) &= C_{p - 1, 0} + C_{0, r - 1} - C_{0, 0} - C_{p - 1, r - 1}\\
                       &\le (A_{p - 1, j} + B_{j, 0}) + (A_{0, j'} + B_{j', r - 1}) - (A_{0, j} + B_{j, 0}) - (A_{p - 1, j'} + B_{j', r - 1})\\
                       &= A_{p - 1, j} + A_{0, j'} - A_{0, j} - A_{p - 1, j'}\\
                       &= \deltasigma(A\fragmentco{0}{p}\fragmentcc{j}{j'})\\
                       &\le \deltasigma(A).
    \end{align*}
    The inequality $\deltasigma(C) \le \deltasigma(B)$ can be obtained using a symmetric argument.
\end{proof}

The following corollary says that every core element of $A \minplus B$ can be attributed to at least one core element of $A$ and at least one core element of $B$.

\begin{corollary}\label{cor:local-core-sum-preservation}
    Let $A$ be a $p \times q$ Monge matrix, let $B$ be a $q \times r$ Monge matrix, and let $C \coloneqq A \minplus B$.
    Consider integers $i \in \fragmentco{0}{p-1}$ and $k \in \fragmentco{0}{r-1}$ such that $\dens{C}_{i, k} \neq 0$.
    There exist integers $j_A, j_B \in \fragmentco{\wit{A}{B}_{i, k}}{\wit{A}{B}_{i + 1, k + 1}}$ such that $\dens{A}_{i, j_A} \neq 0$ and $\dens{B}_{j_B, k} \neq 0$.
\end{corollary}

\begin{proof}
    Let $j = \wit{A}{B}_{i, k}$ and $j' = \wit{A}{B}_{i + 1, k + 1}$.
    By monotonicity of $\wit{A}{B}$ (\cref{lm:witness-monotonicity}), we have $j \le \wit{A}{B}_{i, k + 1}, \wit{A}{B}_{i + 1, k}\le j'$.
    Thus, $C\fragmentcc{i}{i + 1}\fragmentcc{k}{k + 1} = A\fragmentcc{i}{i + 1}\fragmentcc{j}{j'} \minplus B\fragmentcc{j}{j'}\fragmentcc{k}{k + 1}$.
    Due to \cref{lm:core-sum-preservation}, we have $0 < \dens{C}_{i, k} = \deltasigma(C\fragmentcc{i}{i + 1}\fragmentcc{k}{k + 1}) \le \deltasigma(A\fragmentcc{i}{i + 1}\fragmentcc{j}{j'})$.
    Thus, there exists a core element in $\dens{A}\fragmentco{i}{i + 1}\fragmentco{j}{j'}$.
    
    Symmetrically, $0 < \dens{C}_{i, k}= \deltasigma(C\fragmentcc{i}{i + 1}\fragmentcc{k}{k + 1})\le \deltasigma(B\fragmentcc{j}{j'}\fragmentcc{k}{k + 1})$, and there exists a core element in $\dens{B}\fragmentco{j}{j'}\fragmentco{k}{k + 1}$.
\end{proof}

We now use \cref{cor:local-core-sum-preservation} to derive a bound on $\delta(A \minplus B)$ in terms of $\delta(A)$ and $\delta(B)$.
The underlying idea is to show that every core element of $C$ is either the first or the last one (in the lexicographic ordering) that the mapping of \cref{cor:local-core-sum-preservation}
attributes to the corresponding core element of $A$ \emph{or} $B$ (see \cref{fig:core-preservation-proof} for why the opposite would lead to a contradiction).

\lmcorepreservation*
    \begin{figure}[htb]
        \begin{center}
            \includegraphics[scale=1.5]{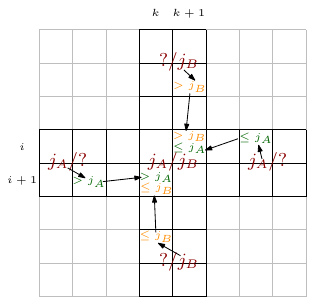}
        \end{center}
    
        \caption{We consider the matrix of witnesses $\wit{A}{B}$, drawn as an array with rows in $[0\dd p)$ indexed form top to bottom and columns in $[0\dd r)$ indexed from left to right.
        In the cells, we write bounds on the corresponding witnesses. 
        In the grid nodes, we write the two values in $[0\dd q-1)$ derived from \cref{cor:local-core-sum-preservation} for each core element of $C$. 
        The center of the picture corresponds to a core value $\dens{C}_{i,k}$, attributed to $\dens{A}_{i,j_A}$ and $\dens{B}_{j_B,k}$, respectively.
        For a proof by contradiction, suppose that there are some core elements to the left and right of $(i, k)$ attributed to $\dens{A}_{i,j_A}$ and some core elements above and below $(i, k)$ attributed to $\dens{B}_{j_B,k}$. 
        The arrows represent implications (based on \cref{cor:local-core-sum-preservation,lm:witness-monotonicity}), from which a contradiction follows: $\textcolor{darkgreen}{j_A <{}} \wit{A}{B}_{i + 1, k} \textcolor{orange}{{}\le j_B <{}}\wit{A}{B}_{i, k + 1} \textcolor{darkgreen}{{}\le j_A}$.}
        \label{fig:core-preservation-proof}
    \end{figure}
    \begin{proof}
        Let $C \coloneqq A \minplus B$.
        Define a function $f_A : \core(C) \to \core(A)$ that maps every core element $(i, k, \dens{C}_{i, k}) \in \core(C)$ to $(i, j, \dens{A}_{i, j})$ for the smallest $j \in \fragmentco{\wit{A}{B}_{i, k}}{\wit{A}{B}_{i + 1, k + 1}}$ with $\dens{A}_{i, j} \neq 0$; such a $j$ exists due to \cref{cor:local-core-sum-preservation}.
        Analogously, we define a function $f_B : \core(C) \to \core(B)$ that maps every core element $(i, k, \dens{C}_{i, k}) \in \core(C)$ to $(j, k, \dens{B}_{j, k})$ for the smallest $j \in \fragmentco{\wit{A}{B}_{i, k}}{\wit{A}{B}_{i + 1, k + 1}}$ with $\dens{B}_{j, k} \neq 0$; again, such a $j$ exists due to \cref{cor:local-core-sum-preservation}.

        \begin{claim}
            Every core element $c \in \core(C)$ is the lexicographically minimal or maximal one in its pre-image $f_A^{-1}(f_A(c))$ under $f_A$ or its pre-image $f_B^{-1}(f_B(c))$ under $f_B$.
        \end{claim}
        \begin{claimproof}
            For a proof by contradiction, pick an element $c\coloneqq (i, k, \dens{C}_{i, k})\in \core(C)$ violating the claim.
            Let $(i, j_A, \dens{A}_{i, j_A}) \coloneqq a \coloneqq f_A(c)$ and $ (j_B, k, \dens{B}_{j_B, k}) \coloneqq b \coloneqq f_B(c)$.
            We make four symmetric arguments (all illustrated in \cref{fig:core-preservation-proof}), with the first one presented in more detail. 
            \begin{enumerate}
                \item Since $c$ is not the minimal core element in $f_A^{-1}(a)$, we have $f_A(c')=a$ for some core element $(i', k', \dens{C}_{i', k'})\coloneqq c'\in \core(C)$ that precedes $c$ in the lexicographical order, denoted $c' \prec c$.
                Due to $f_A(c')=a=(i, j_A, \dens{A}_{i, j_A})$, we must have $i'=i$ and $j_A < \wit{A}{B}_{i + 1, k' + 1}$. 
                Since $i'=i$, then $c'\prec c$ implies $k'+1\le k$.
                The monotonicity of the witness matrix (\cref{lm:witness-monotonicity}) thus yields $\wit{A}{B}_{i + 1, k' + 1} \le \wit{A}{B}_{i+1,k}$.
                Overall, we conclude that $j_A < \wit{A}{B}_{i+1,k}$.
                \item Since $c$ is not the maximal core element in $f_A^{-1}(a)$, we have $f_A(c')=a$ for some core element $c' \succ c$. In particular, $c'=(i, k', \dens{C}_{i, k'})\in \core(C)$ for some $k'>k$.
                Then, $\wit{A}{B}_{i, k+1} \le \wit{A}{B}_{i, k'} \le j_A$ follows from \cref{lm:witness-monotonicity} and $f_A(c')=a$, respectively.
                \item Since $c$ is not the minimal core element in $f_B^{-1}(b)$, we have $f_B(c')=b$ for some core element $c' \prec c$. In particular, $c'=(i', k, \dens{C}_{i', k})\in \core(C)$ for some $i'<i$.
                Then, $\wit{A}{B}_{i, k+1} \ge \wit{A}{B}_{i'+1, k+1} > j_B$ follows from \cref{lm:witness-monotonicity} and $f_B(c')=b$, respectively.
                \item Since $c$ is not the maximal core element in $f_B^{-1}(b)$, we have $f_B(c')=b$ for some core element $c' \succ c$. In particular, $c'=(i', k, \dens{C}_{i', k})\in \core(C)$ for some $i'>i$.
                Then, $\wit{A}{B}_{i+1, k} \le \wit{A}{B}_{i', k} \le j_B$ follows from \cref{lm:witness-monotonicity} and $f_B(c')=b$, respectively.
            \end{enumerate}
            Overall, we derive a contradiction: $j_A < \wit{A}{B}_{i + 1, k} \le j_B < \wit{A}{B}_{i, k + 1} \le j_A$.
        \end{claimproof}
        As every core element of $C$ is either the first or the last one in some pre-image under $f_A$ or $f_B$, we get that $\delta(C) \le 2 \cdot (\delta(A) + \delta(B))$.
\end{proof}

\begin{example}\label{ex:core-preservation}
    Note that the inequality for the core size from \cref{lm:core-preservation} is weaker than the inequality for the core sum from \cref{lm:core-sum-preservation}.
    We claim that this weakening is not an artifact of our proof.
    Consider the following Monge matrices
    \[
        A=
          \begin{pmatrix}
                0 & 4 & 5 & 6\\
                0 & 1 & 2 & 3\\
                0 & 1 & 2 & 0\\
                0 & 1 & 2 & 0
          \end{pmatrix},
  \qquad
        B=
          \begin{pmatrix}
                0 & 2 & 4 & 6\\
                0 & 0 & 2 & 4\\
                0 & 0 & 0 & 2\\
                0 & 0 & 0 & 0
          \end{pmatrix},
    \qquad
        C=
          \begin{pmatrix}
                0 & 2 & 4 & 6\\
                0 & 1 & 2 & 3\\
                0 & 0 & 0 & 0\\
                0 & 0 & 0 & 0
          \end{pmatrix}.
    \]

    We have $C = A \minplus B$, $\delta(A) = 2$, $\delta(B) = 3$, and $6 = \delta(C) > \delta(A) + \delta(B) > \min\{\delta(A), \delta(B)\}$.
    %\lipicsEnd
\end{example}

\begin{remark}
    To the best of our knowledge, \cref{lm:core-preservation} shows the first bound on the core size of $A \minplus B$ in terms of the core sizes of $A$ and $B$.
    Previously, such a bound was only known for unit-Monge matrices \cite{Tiskin10}.
    In particular, our bound allows simplifying some time complexities of already existing algorithms.
    For example, \cite[Lemma 18]{Russo10} shows how to compute the product of two explicitly given Monge matrices $A$ and $B$ in time $\Oh(p \log q + q \log r + \delta(A \minplus B) \log q \log r + (\delta(A) + \delta(B)) \log^2 q)$. 
    Assuming that $r \in \poly(q)$, \cref{lm:core-preservation} tells us that we can drop the third summand in the time complexity.
    %\lipicsEnd
\end{remark}

\section{Core-Sparse Monge Matrix Min-Plus Multiplication Algorithm}\label{sec:core-based-multiplication-algorithm}

In this section, we present an algorithm that, given the condensed representations of two Monge matrices, computes the condensed representation of their $(\min, +)$-product.
Even though the condensed representation itself does not provide efficient random access to the values of the matrix, the following fact justifies our choice of the representation of Monge matrices.

\begin{fact}[{\cite[Lemma 4.4]{GK24}}] \label{lm:cmo}
    There exists an algorithm that, given the condensed representation of a matrix $A$ of size $p \times q$, in time $\Oh(p + q + \delta(A) \log (1 + \delta(A)))$ builds a \emph{core-based matrix oracle data structure} $\mds(A)$ that provides $\Oh(\log (2 + \delta(A)))$-time random access to the entries of~$A$.

\end{fact}

\begin{proof}[Proof Sketch]
    By \cref{lm:calculate-value-from-core} applied for $a=c=0$, it suffices to answer orthogonal range sum queries on top of $\core(A)$.
    According to \cite[Theorem 3]{Wil85}, it can be done in $\Oh(\delta(A) \log (1 + \delta(A)))$ preprocessing time and $\Oh(\log (2 + \delta(A)))$ query time.
\end{proof}

\begin{definition}
    Let $A$ be a $p \times q$ matrix.
    For every $i \in \fragmentco{0}{p-1}$, denote $\corei{i}(A) = \{(i, j, \dens{A}_{i, j}) \mid j \in \fragmentco{0}{q-1}, \dens{A}_{i, j} \neq 0 \}$ and $\deltai{i}(A) = |\corei{i}(A)|$.
    Analogously, for every $j \in \fragmentco{0}{q-1}$, denote $\corej{j}(A) = \{(i, j, \dens{A}_{i, j}) \mid i \in \fragmentco{0}{p-1}, \dens{A}_{i, j} \neq 0 \}$ and $\deltaj{j}(A) = |\corej{j}(A)|$.
%\lipicsEnd
\end{definition}

\begin{observation}
    $\delta(A) = \sum_{i \in \fragmentco{0}{p-1}} \deltai{i}(A) = \sum_{j \in \fragmentco{0}{q-1}} \deltaj{j}(A)$.
%\lipicsEnd
\end{observation}

We now design another matrix oracle as an alternative to $\mds(A)$ for the situation in which one wants to query an entry of $A$ that is adjacent to some other already known entry.

\begin{lemma}\label{lm:local-core-oracle}
    There is an algorithm that, given the condensed representation of a $p \times q$ matrix $A$, in time $\Oh(p + q + \delta(A))$ builds a \emph{local core oracle data structure} $\lco(A)$ with the following interface.
    \begin{description}
        \item[Boundary access:] given indices $i \in \fragmentco{0}{p}$ and $j \in \fragmentco{0}{q}$ such that $i=0$ or $j=0$, in time $\Oh(1)$ returns $A_{i, j}$.
        \item[Vertically adjacent recomputation:] given indices $i \in \fragmentco{0}{p-1}$ and $j \in \fragmentco{0}{q}$, in time $\Oh(\deltai{i}(A) + 1)$ returns $A_{i+1,j}-A_{i, j}$.
        \item[Horizontally adjacent recomputation:] given indices $i \in \fragmentco{0}{p}$ and $j \in \fragmentco{0}{q-1}$, in time \mbox{$\Oh(\deltaj{j}(A) + 1)$} returns $A_{i,j+1}-A_{i, j}$.
    \end{description}
\end{lemma}

\begin{proof}
    The local core oracle data structure of $A$ stores all the values of the topmost row and the leftmost column of $A$, as well as two collections of lists: $\corei{i}(A)$ for all $i \in \fragmentco{0}{p-1}$ and $\corej{j}(A)$ for all $j \in \fragmentco{0}{q-1}$.
    The values of the topmost row and the leftmost column of $A$ are already given.
    The lists $\corei{i}(A)$ and $\corej{j}(A)$ can be computed in $\Oh(p + q + \delta(A))$ time from $\core(A)$.
    Hence, we can build $\lco(A)$ in time $\Oh(p + q + \delta(A))$.
    
    Boundary access can be implemented trivially.
    We now show how to implement the vertically adjacent recomputation.
    Suppose that we are given $i \in \fragmentco{0}{p-1}$ and $j \in \fragmentco{0}{q}$.
    Due to \cref{lm:calculate-value-from-core}, we have $A_{i, j}+ A_{i + 1, 0} = A_{i + 1, j} + A_{i, 0} + \deltasigma(A\fragmentcc{i}{i + 1}\fragmentcc{0}{j})$.
    Note that the values $A_{i, 0}$ and $A_{i + 1, 0}$ can be computed in constant time using boundary access queries, and $\deltasigma(A\fragmentcc{i}{i + 1}\fragmentcc{0}{j})$ can be computed from $\corei{i}(A)$ in time $\Oh(\deltai{i}(A) + 1)$.
    Hence, $A_{i+1, j}-A_{i,j}$ can be computed in time $\Oh(\deltai{i}(A) + 1)$.

    The horizontally adjacent recomputation is implemented analogously so that $A_{i, j+1}-A_{i,j}$ can be computed in time $\Oh(\deltaj{j}(A) + 1)$.
\end{proof}

Note that vertically and horizontally adjacent recomputations allow computing the neighbors of any given entry $A_{i,j}$.

\begin{lemma}\label{lm:get-submatrix}
    Given the condensed representation of a $p \times q$ matrix $A$, the condensed representation of any contiguous submatrix $A'$ of $A$ can be computed in time $\Oh(p + q + \delta(A))$.
\end{lemma}

\begin{proof}
    Say, $A' = A\fragmentcc{i}{i'}\fragmentcc{j}{j'}$.
    Note that $\core(A')$ can be obtained in time $\Oh(\delta(A))$ by filtering out all elements of $\core(A)$ that lie outside of $\fragmentco{i}{i'}\times \fragmentco{j}{j'}$.
    It remains to obtain the topmost row and the leftmost column of $A'$.
    In time $\Oh(p + q + \delta(A))$ we build $\lco(A)$ of \cref{lm:local-core-oracle}.
    By starting from $A_{i, 0}$ and repeatedly applying the horizontally adjacent recomputation, we can compute $A\fragmentcc{i}{i}\fragmentco{0}{q}$ (and thus $A\fragmentcc{i}{i}\fragmentcc{j}{j'}$) in time $\Oh(\sum_{j \in \fragmentco{0}{q}} (\deltaj{j}(A) + 1)) = \Oh(q + \delta(A))$.
    The leftmost column of $A'$ can be computed in time $\Oh(p + \delta(A))$ analogously.
\end{proof}

We now show a helper lemma that compresses ``ultra-sparse'' Monge matrices to limit their dimensions by the size of their core.

\begin{restatable}[Matrix compression]{lemma}{lmcompression} \label{lm:compression}
    \SetKwFunction{compress}{Compress}
    \SetKwFunction{decompress}{Decompress}
    There are two algorithms \compress and \decompress with the following properties:
    The \compress algorithm, given the condensed representations of a $p^* \times q^*$ Monge matrix $A^*$ and a $q^* \times r^*$ Monge matrix $B^*$, in time $\Oh(p^* + q^* + r^* + \delta(A^*) + \delta(B^*))$, builds a $p \times q$ Monge matrix $A$ and a $q \times r$ Monge matrix $B$ such that $p \le \delta(A^*) + 1$, $q \le \delta(A^*) + \delta(B^*)+1$, $r \le \delta(B^*) + 1$, $\delta(A) = \delta(A^*)$, and $\delta(B) = \delta(B^*)$.
    The \decompress algorithm, given the condensed representations of $A^*$, $B^*$, and $A \minplus B$, where $(A, B) = \compress(A^*, B^*)$, computes the condensed representation of $A^* \minplus B^*$ in time $\Oh(p^* + r^* + \delta(A^*) + \delta(B^*))$.
\end{restatable}

\begin{proof}[Proof Sketch]
    The \compress algorithm deletes all rows of $A^*$ that do not contain any core elements and all columns of $B^*$ that do not contain any core elements.
    This way, $p \le \delta(A^*) + 1$ and $r \le \delta(B^*) + 1$ are guaranteed.
    Furthermore, all core elements of these two matrices are preserved.
    If there are no core elements between rows $i$ and $i+1$ of $A^*$, the corresponding entries of these two rows differ by some constant $c$, and thus all the entries of the row $i+1$ of $A^* \otimes B^*$ can be obtained from the corresponding entries of row $i$ of $A^* \otimes B^*$ by adding~$c$.
    Since we can recover $c$ from $A^*$, no information about the answer is ``lost'' when deleting such rows.
    An analogous property holds for the removed columns of $B^*$.
    The \decompress algorithm reverses the row and column removals performed by the \compress algorithm.

    The \compress algorithm also reduces the number $q^*$ of columns of $A^*$ and rows of $B^*$ to $q \le \delta(A^*) + \delta(B^*) + 1$.
    Observe that if there are no core elements between columns $j$ and $j + 1$ of $A^*$ and between rows $j$ and $j + 1$ of $B^*$, then the values $A^*_{i,j}+B^*_{j,k}$ and $A^*_{i,j+1}+B^*_{j+1,k}$ differ by a constant independent of $i$ and $k$.
    Depending on the sign of this constant difference, we can delete either the $j$-th column of $A^*$ and the $j$-th row of $B^*$ or the $(j+1)$-th column of $A^*$ and the $(j+1)$-th row of $B^*$ so that we do not change the min-plus product of these matrices.
    By repeating this process for all indices $j$ that do not ``contribute'' to the cores of $A^*$ and $B^*$, we get $q \le \delta(A^*) + \delta(B^*) + 1$.
    The formal proof is deferred to \cref{sec:appendix-lmcompression}.
\end{proof}

Finally, we show our main Monge matrix multiplication algorithm.

\lmcorebasedmultiplicationalgorithm*

\begin{proof}[Proof Sketch]
    \SetKwFunction{recmult}{Multiply}

    We design a recursive divide-and-conquer procedure $\recmult(A^*, B^*)$ (see \cref{alg:core-based-multiplication-algorithm} for the details) and solve the problem by initially running $\recmult(M_1, M_2)$.
    Given the matrices $A^*$ and $B^*$, we first apply the \compress algorithm of \cref{lm:compression} to compress $A^*$ and $B^*$ into a $p \times q$ matrix $A$ and a $q \times r$ matrix $B$ respectively.
    We then compute the condensed representation of $C \coloneqq A \minplus B$ and finally use the \decompress algorithm of \cref{lm:compression} to decompress $C$ into $A^* \minplus B^*$.

    If $p = r = 1$, we compute the matrix $C$ trivially.
    Otherwise, we pick a splitting point $m \in \fragmentoo{0}{q}$, split the matrix $A$ vertically into the matrix $A^L$ of size $p \times m$ and the matrix $A^R$ of size $p \times (q - m)$, and split the matrix $B$ horizontally into the matrix $B^L$ of size $m \times r$ and the matrix $B^R$ of size $(q - m) \times r$.
    We pick $m$ in such a way that it splits the cores of $A$ and $B$ almost equally across $A^L$ and $B^L$ and $A^R$ and $B^R$, that is, $\delta(A^L) + \delta(B^L)$ and $\delta(A^R) + \delta(B^R)$ are at most $(\delta(A) + \delta(B)) / 2$.
    We recursively compute the condensed representations of the matrices $C^L \coloneqq A^L \minplus B^L$ and $C^R \coloneqq A^R \minplus B^R$.
    The resulting matrix $C$ can be obtained as the element-wise minimum of $C^L$ and $C^R$.
    Furthermore, one can see that, due to \cref{lm:witness-monotonicity}, in some top-left region of~$C$, the values are equal to the corresponding values of $C^L$, and in the remaining bottom-right region of $C$, the values are equal to the corresponding values of $C^R$; see \cref{fig:monge-multuplication-conquer-step}.
    We find the boundary between these two regions by starting in the bottom-left corner of $C$ and traversing it towards the top-right corner along this boundary.
    We use $\lco(C^L)$ and $\lco(C^R)$ to sequentially compute the subsequent entries along the boundary.
    Having determined the boundary, we construct the core of $C$ by picking the core elements of $C^L$ located to the top-left of this boundary, picking the core elements of $C^R$ located to the bottom-right of this boundary, and computing the $\dens{C}$ values on the boundary from scratch.
    The values in the topmost row and the leftmost column of $C$ can be trivially obtained as element-wise minima of the corresponding values in $C^L$ and $C^R$.
    This concludes the recursive procedure; see \cref{alg:core-based-multiplication-algorithm} for the pseudocode.
    This algorithm follows the classical divide-and-conquer framework, and thus its time complexity can be easily derived from \cref{lm:core-preservation}.
    The full description of the algorithm, the proof of its correctness, and the formal analysis of its time complexity are deferred to \cref{sec:appendix-lmcorebasedmultiplicationalgorithm}.
\end{proof}
\SetCommentSty{mycommentfont}
\newcommand{\mycommentfont}[1]{\textcolor{gray}{#1}}
\begin{algorithm}
    \caption{The algorithm from \cref{lm:core-based-multiplication-algorithm}. Given the condensed representations of the Monge matrices $A^*$ and $B^*$, the algorithm returns the condensed representation of $A^* \minplus B^*$.} \label{alg:core-based-multiplication-algorithm}
\recmult{$A^*, B^*$} \Begin{
    $A, B \gets \compress(A^*, B^*)$\;
    $p \gets$ height of $A$; $q \gets$ width of $A$; $r \gets$ width of $B$\;
    \If{$p = 1$ \KwSty{and} $r = 1$}{
        $C \gets 1 \times 1$ matrix with the entry equal to $\min\{A_{0, j} + B_{j, 0} \mid j \in \fragmentco{0}{q}\}$\;
        \Return{$\decompress(A^*, B^*, C)$}\;
    }
    Sort $\core(A)$ by columns and $\core(B)$ by rows\tcp*{needed to find $m$}
    Pick the largest $m \in \fragmentoo{0}{q}$ with $\delta(A\fragmentco{0}{p}\fragmentco{0}{m}) + \delta(B\fragmentco{0}{m}\fragmentco{0}{r}) \le (\delta(A) + \delta(B)) / 2$\;
    Build the condensed representations of $A\fragmentco{0}{p}\fragmentco{0}{m}$ and $B\fragmentco{0}{m}\fragmentco{0}{r}$\;
    $C^L \gets \recmult(A\fragmentco{0}{p}\fragmentco{0}{m}, B\fragmentco{0}{m}\fragmentco{0}{r})$; Build $\lco(C^L)$ using \cref{lm:local-core-oracle}\;
    Build the condensed representations of $A\fragmentco{0}{p}\fragmentco{m}{q}$ and $B\fragmentco{m}{q}\fragmentco{0}{r}$\;
    $C^R \gets \recmult(A\fragmentco{0}{p}\fragmentco{m}{q}, B\fragmentco{m}{q}\fragmentco{0}{r})$; Build $\lco(C^R)$ using \cref{lm:local-core-oracle}\;
    $lst \gets$ list of length $p$\tcp*{vertical edges of the ladder from \cref{fig:monge-multuplication-conquer-step}}
    $i \gets p - 1, j \gets -1$\tcp*{current matrix coordinates while traversing the ladder}
    Maintain the well-defined values in $C^L\fragmentcc{i}{i+1}\fragmentcc{j}{j+1}$ and $C^R\fragmentcc{i}{i+1}\fragmentcc{j}{j+1}$ using $\lco(C^L)$ and $\lco(C^R)$, and the values in $C\fragmentcc{i}{i+1}\fragmentcc{j}{j+1}$ as element-wise minima\;
    $Core \gets \{\}$\;
    \While{$i \ge 0$}{
        \If{$i < p - 1$ \KwSty{and} $0 \le j < r - 1$ \KwSty{and} $C_{i, j + 1} + C_{i + 1, j} \neq C_{i, j} + C_{i + 1, j + 1}$}{
            $Core \gets Core \cup \{(i, j, C_{i, j + 1} + C_{i + 1, j} - C_{i, j} - C_{i + 1, j + 1})\}$\;
        }
        \If{$j = r - 1$ \KwSty{or} $C^L_{i, j + 1} > C^R_{i, j + 1}$}{
            $lst_i \gets j$\;
            $i \gets i - 1$\;
        }\Else{
            $j \gets j + 1$\;
        }
    }
    $Core \gets Core \cup \{(i, j, C^{L\square}_{i, j}) \in \core(C^L) \mid j < lst_{i + 1}\} \cup \{(i, j, C^{R\square}_{i, j}) \in \core(C^R) \mid j > lst_i\}$\;
    Compute the topmost row and the leftmost column of $C$ as the element-wise minima of the topmost rows and the leftmost columns of $C^L$ and $C^R$\;
    Condensed representation of $C \gets (C\fragmentcc{0}{0}\fragmentco{0}{r}, C\fragmentco{0}{p}\fragmentcc{0}{0}, Core)$\;
    \Return{$\decompress(A^*, B^*, C)$}\;
}
\end{algorithm}

    \begin{figure}
        \begin{center}
            \begin{tikzpicture}
                \input{monge-multiplication-conquer-step.tex}

                \node[darkred,below] at (11.5, 2) {\Large $C^R$};
                \node[darkgreen,below] at (1.5, 2) {\Large $C^L$};
                \node[below] at (6.5, 2) {\Large $C$};

                \draw[-latex, thick] (3.05, 3.5) -- (4.95, 3.5);
                \draw[-latex, thick] (9.95, 3.5) -- (8.05, 3.5);

            \end{tikzpicture}
        \end{center}
    
        \caption{An example of how $C$ is obtained from $C^L$ and $C^R$. The blue ladder represents the border between the values that are inherited from $C^L$ and the values that are inherited from $C^R$. Ticks correspond to the core elements: the green ticks are inherited from $C^L$, the red ticks are inherited from $C^R$, and the blue ticks represent the core element computed from scratch.}
        \label{fig:monge-multuplication-conquer-step}
    \end{figure}

If the $(\min, +)$-product of matrices represents the lengths of the shortest paths in a graph, the corresponding witness matrix allows computing the underlying shortest paths themselves.
For that, we extend the presented algorithm to allow witness reconstruction.

\begin{restatable}{theorem}{lmproductwitnessreconstruction} \label{lm:product-witness-reconstruction}
    The algorithm of \cref{lm:core-based-multiplication-algorithm} can be extended so that, within the same time complexity, it also builds a data structure that takes $\Oh(n_1 + n_2 + n_3 + \delta(M_1) + \delta(M_2))$ space and provides $\Oh(\log (2+\delta(M_1) + \delta(M_2)))$-time oracle access to $\wit{M_1}{M_2}$.
\end{restatable}

\begin{proof}[Proof Sketch]
    We slightly modify the algorithm of \cref{lm:core-based-multiplication-algorithm}.
    For the leaf recursive calls with $p = r = 1$, we explicitly store the minimal witness of the only entry of $C$.
    In the non-leaf recursive calls, we store the correspondence between the indices in the compressed matrices $A, B, C$ and the decompressed matrices $A^*, B^*, C^*$, as well as the border separating the values of $C^L$ and $C^R$ in $C$ (the blue curve of \cref{fig:monge-multuplication-conquer-step}). 
    Naively, this data takes space proportional to the time complexity of \cref{lm:core-based-multiplication-algorithm}.
    Nevertheless, using bit-masks equipped with rank/select functionality, the total space complexity can be brought down to be linear in terms of the size of the input.
    To answer a query, we descend the recursion tree.
    In constant time we can find an entry of $C$ corresponding to the queried entry of $C^*$ and decide on which side of the border separating the values of $C^L$ and $C^R$ the queried entry of $C$ is.
    After that, we recurse into either $A^L \otimes B^L$ or $B^L \otimes B^R$.
    In the terminal recursion call with $p = r = 1$ we return the minimal witness that is stored explicitly.
    The time complexity of the algorithm is proportional to the depth of recursion.
    See \cref{sec:appendix-lmproductwitnessreconstruction} for a formal description of the algorithm.
\end{proof}

\section{Range LIS Queries: Sketch}\label{sec:range-lis}
One of the original applications of Tiskin's procedure for simple unit-Monge matrix multiplication~\cite{Tis07,Tiskin10} is an algorithm that preprocesses a given sequence $(s\position{i})_{i \in \fragmentco{0}{n}}$ of integers in $\Oh(n \log^2 n)$ time so that \emph{Range Longest Increasing Subsequence} (Range LIS) queries can be answered in $\Oh(\log n)$ time.
We show an alternative way of obtaining the same result using \cref{lm:core-based-multiplication-algorithm} and avoiding the seaweed braid formalism.
The extension of \cref{lm:product-witness-reconstruction} allows recovering the underlying longest increasing subsequence of length $\ell$ in $\Oh(\ell \log^2 n)$ time. 
Further $\Oh(n \log^3 n)$-time preprocessing allows for $\Oh(\ell)$-time recovery, which improves upon the result of \cite{KR24}.

Without loss of generality we assume that $(s\position{i})_{i \in \fragmentco{0}{n}}$ is a permutation of $\fragmentco{0}{n}$.\footnote{In what follows, we reserve the word ``permutation'' for permutations of $\fragmentco{0}{m}$ for some $m\in \mathbb{Z}_+$.}
We use a popular tool for string similarity problems and interpret range LIS queries as computing the longest path between some pair of vertices of a corresponding \emph{alignment graph} $G^s$; see \cref{fig:alignment-graph-simple}.
The crucial property of $G^s$ is that it is planar, and thus due to \cite[Section 2.3]{FR06}, the matrix $M^s$ of longest distances from the vertices in the bottom row to the vertices of the top row of this graph is anti-Monge.\footnote{Note that, in reality, some entries of this matrix are infinite. In the formal description of the algorithm, we augment the alignment graph so that the finite entries of this matrix are preserved and the infinite entries become finite.}
Thus, the problem of constructing a data structure for range LIS queries is reduced to the problem of computing the condensed representation of~$M^s$.

\begin{figure}
    \begin{center}
        \begin{tikzpicture}[scale=0.75]
            \node[inner sep = 1pt,circle,fill=black] (n0_0) at (0,0) {};
\node[inner sep = 1pt,circle,fill=black] (n0_1) at (0,1) {};
\draw[-latex,darkred] (n0_0) -- (n0_1);
\node[inner sep = 1pt,circle,fill=black] (n0_2) at (0,2) {};
\draw[-latex,darkred] (n0_1) -- (n0_2);
\node[inner sep = 1pt,circle,fill=black] (n0_3) at (0,3) {};
\draw[-latex,darkred] (n0_2) -- (n0_3);
\node[inner sep = 1pt,circle,fill=black] (n0_4) at (0,4) {};
\draw[-latex,darkred] (n0_3) -- (n0_4);
\node[inner sep = 1pt,circle,fill=black] (n0_5) at (0,5) {};
\draw[-latex,darkred] (n0_4) -- (n0_5);
\node[inner sep = 1pt,circle,fill=black] (n1_0) at (1,0) {};
\draw[-latex,darkred] (n0_0) -- (n1_0);
\node[inner sep = 1pt,circle,fill=black] (n1_1) at (1,1) {};
\draw[-latex,darkred] (n0_1) -- (n1_1);
\draw[-latex,darkred] (n1_0) -- (n1_1);
\node[inner sep = 1pt,circle,fill=black] (n1_2) at (1,2) {};
\draw[-latex,darkred] (n0_2) -- (n1_2);
\draw[-latex,darkred] (n1_1) -- (n1_2);
\draw[-latex,very thick,darkgreen] (n0_1) -- (n1_2);
\node[inner sep = 1pt,circle,fill=black] (n1_3) at (1,3) {};
\draw[-latex,darkred] (n0_3) -- (n1_3);
\draw[-latex,darkred] (n1_2) -- (n1_3);
\node[inner sep = 1pt,circle,fill=black] (n1_4) at (1,4) {};
\draw[-latex,darkred] (n0_4) -- (n1_4);
\draw[-latex,darkred] (n1_3) -- (n1_4);
\node[inner sep = 1pt,circle,fill=black] (n1_5) at (1,5) {};
\draw[-latex,darkred] (n0_5) -- (n1_5);
\draw[-latex,darkred] (n1_4) -- (n1_5);
\node[inner sep = 1pt,circle,fill=black] (n2_0) at (2,0) {};
\draw[-latex,darkred] (n1_0) -- (n2_0);
\node[inner sep = 1pt,circle,fill=black] (n2_1) at (2,1) {};
\draw[-latex,darkred] (n1_1) -- (n2_1);
\draw[-latex,darkred] (n2_0) -- (n2_1);
\draw[-latex,very thick,darkgreen] (n1_0) -- (n2_1);
\node[inner sep = 1pt,circle,fill=black] (n2_2) at (2,2) {};
\draw[-latex,darkred] (n1_2) -- (n2_2);
\draw[-latex,darkred] (n2_1) -- (n2_2);
\node[inner sep = 1pt,circle,fill=black] (n2_3) at (2,3) {};
\draw[-latex,darkred] (n1_3) -- (n2_3);
\draw[-latex,darkred] (n2_2) -- (n2_3);
\node[inner sep = 1pt,circle,fill=black] (n2_4) at (2,4) {};
\draw[-latex,darkred] (n1_4) -- (n2_4);
\draw[-latex,darkred] (n2_3) -- (n2_4);
\node[inner sep = 1pt,circle,fill=black] (n2_5) at (2,5) {};
\draw[-latex,darkred] (n1_5) -- (n2_5);
\draw[-latex,darkred] (n2_4) -- (n2_5);
\node[inner sep = 1pt,circle,fill=black] (n3_0) at (3,0) {};
\draw[-latex,darkred] (n2_0) -- (n3_0);
\node[inner sep = 1pt,circle,fill=black] (n3_1) at (3,1) {};
\draw[-latex,darkred] (n2_1) -- (n3_1);
\draw[-latex,darkred] (n3_0) -- (n3_1);
\node[inner sep = 1pt,circle,fill=black] (n3_2) at (3,2) {};
\draw[-latex,darkred] (n2_2) -- (n3_2);
\draw[-latex,darkred] (n3_1) -- (n3_2);
\node[inner sep = 1pt,circle,fill=black] (n3_3) at (3,3) {};
\draw[-latex,darkred] (n2_3) -- (n3_3);
\draw[-latex,darkred] (n3_2) -- (n3_3);
\node[inner sep = 1pt,circle,fill=black] (n3_4) at (3,4) {};
\draw[-latex,darkred] (n2_4) -- (n3_4);
\draw[-latex,darkred] (n3_3) -- (n3_4);
\draw[-latex,very thick,darkgreen] (n2_3) -- (n3_4);
\node[inner sep = 1pt,circle,fill=black] (n3_5) at (3,5) {};
\draw[-latex,darkred] (n2_5) -- (n3_5);
\draw[-latex,darkred] (n3_4) -- (n3_5);
\node[inner sep = 1pt,circle,fill=black] (n4_0) at (4,0) {};
\draw[-latex,darkred] (n3_0) -- (n4_0);
\node[inner sep = 1pt,circle,fill=black] (n4_1) at (4,1) {};
\draw[-latex,darkred] (n3_1) -- (n4_1);
\draw[-latex,darkred] (n4_0) -- (n4_1);
\node[inner sep = 1pt,circle,fill=black] (n4_2) at (4,2) {};
\draw[-latex,darkred] (n3_2) -- (n4_2);
\draw[-latex,darkred] (n4_1) -- (n4_2);
\node[inner sep = 1pt,circle,fill=black] (n4_3) at (4,3) {};
\draw[-latex,darkred] (n3_3) -- (n4_3);
\draw[-latex,darkred] (n4_2) -- (n4_3);
\draw[-latex,very thick,darkgreen] (n3_2) -- (n4_3);
\node[inner sep = 1pt,circle,fill=black] (n4_4) at (4,4) {};
\draw[-latex,darkred] (n3_4) -- (n4_4);
\draw[-latex,darkred] (n4_3) -- (n4_4);
\node[inner sep = 1pt,circle,fill=black] (n4_5) at (4,5) {};
\draw[-latex,darkred] (n3_5) -- (n4_5);
\draw[-latex,darkred] (n4_4) -- (n4_5);
\node[inner sep = 1pt,circle,fill=black] (n5_0) at (5,0) {};
\draw[-latex,darkred] (n4_0) -- (n5_0);
\node[inner sep = 1pt,circle,fill=black] (n5_1) at (5,1) {};
\draw[-latex,darkred] (n4_1) -- (n5_1);
\draw[-latex,darkred] (n5_0) -- (n5_1);
\node[inner sep = 1pt,circle,fill=black] (n5_2) at (5,2) {};
\draw[-latex,darkred] (n4_2) -- (n5_2);
\draw[-latex,darkred] (n5_1) -- (n5_2);
\node[inner sep = 1pt,circle,fill=black] (n5_3) at (5,3) {};
\draw[-latex,darkred] (n4_3) -- (n5_3);
\draw[-latex,darkred] (n5_2) -- (n5_3);
\node[inner sep = 1pt,circle,fill=black] (n5_4) at (5,4) {};
\draw[-latex,darkred] (n4_4) -- (n5_4);
\draw[-latex,darkred] (n5_3) -- (n5_4);
\node[inner sep = 1pt,circle,fill=black] (n5_5) at (5,5) {};
\draw[-latex,darkred] (n4_5) -- (n5_5);
\draw[-latex,darkred] (n5_4) -- (n5_5);
\draw[-latex,very thick,darkgreen] (n4_4) -- (n5_5);
\draw (0, -0.25) node[below]{$\mathtt{0}$};
\draw (1, -0.25) node[below]{$\mathtt{1}$};
\draw (2, -0.25) node[below]{$\mathtt{2}$};
\draw (3, -0.25) node[below]{$\mathtt{3}$};
\draw (4, -0.25) node[below]{$\mathtt{4}$};
\draw (5, -0.25) node[below]{$\mathtt{5}$};
\draw (-0.25,0) node[left]{$\mathtt{0}$};
\draw (-0.25,1) node[left]{$\mathtt{1}$};
\draw (-0.25,2) node[left]{$\mathtt{2}$};
\draw (-0.25,3) node[left]{$\mathtt{3}$};
\draw (-0.25,4) node[left]{$\mathtt{4}$};
\draw (-0.25,5) node[left]{$\mathtt{5}$};
        \end{tikzpicture}
        \vspace{-.5cm}
    \end{center}

    \caption{The alignment graph for $s = [1, 0, 3, 2, 4]$, with red edges of weight $0$ and green edges of weight $1$. The length of the longest path from the $i$-th vertex in the bottom row to the $j$-th vertex in the top row of this graph for $i < j$ is equal to $\lis(s\fragmentco{i}{j})$.}
    \label{fig:alignment-graph-simple}
\end{figure}

As $\lis(s\fragmentco{i}{j}) \le n$ for all $i, j \in \fragmentcc{0}{n}$ with $i < j$, all entries of $M^s$ are bounded by $\Oh(n)$, and thus $\delta(M^s) = \Oh(n)$ holds due to \cref{lm:simple-core-bound,lm:calculate-value-from-core}.
We compute $M^s$ in a divide-and-conquer fashion.
We split the sequence $s$ into subsequences $s^{\lo}$ and $s^{\hi}$ containing all values in $[0\dd \floor{\frac{n}{2}})$ and $[\floor{\frac{n}{2}}\dd n)$ respectively, and recursively compute $M^{s^{\lo}}$ and $M^{s^{\hi}}$.
After that, we note that $G^{s^{\lo}}$ and $G^{s^{\hi}}$ are essentially compressed versions of the lower half and the upper half of $G^s$, respectively.
Based on this, we transform $M^{s^{\lo}}$ in $\Oh(n)$ time into the matrix of the longest distances from the vertices in the bottom row of $G^s$ to the vertices in the middle row of $G^s$.
Analogously, we transform $M^{s^{\hi}}$ into the matrix of the longest distances from the middle row of $G^s$ to the top row of~$G^s$.
To obtain $M^s$, it remains to $(\max, +)$-multiply these two matrices using \cref{lm:core-based-multiplication-algorithm}.
Every single recursive call takes time $\Oh(n \log n)$ dominated by the algorithm of \cref{lm:core-based-multiplication-algorithm}.
The whole divide-and-conquer procedure takes $\Oh(n \log^2 n)$ time.
Given the condensed representation of $M^s$, we use \cref{lm:cmo} to create an oracle for $\Oh(\log n)$-time range LIS queries, thus replicating the result of~\cite{Tiskin10}.

\medskip

Compared to the results of \cite{Tiskin10}, this algorithm operates directly on the local LIS values and thus can be easily converted into an algorithm for the reporting version of range LIS queries, where we want to not only find the length of the longest path in the alignment graph but to also compute the structure of the underlying path itself.
For that, we simply use the witness reconstruction oracle of \cref{lm:product-witness-reconstruction} to find the midpoint of the path and descend the recursion tree.
This costs $\Oh(\log n)$ time per recursive call and allows reconstructing the entire length-$\ell$ LIS in time $\Oh(\ell \log^2 n)$.
A similar recursive LIS reporting scheme is used in~\cite{CHS23}.

By treating increasing subsequences of length at most $\log^2 n$ separately, we further obtain an algorithm with $\Oh(n \log^3 n)$ preprocessing time and $\Oh(\ell)$ reporting time.
It improves the result of \cite{KR24} with $\Oh(n^{3 / 2} \polylog n)$-time preprocessing and $\Oh(\ell + n^{1 / 2} \polylog n)$-time reporting.

The formal proofs of the results in this section are given in \cref{app:range-lis}.

\bibliography{refs}

\appendix

\section{Deferred Proofs from Sections~\ref{sec:preliminaries} and~\ref{sec:core-based-multiplication-algorithm}}\label{sec:appendix}

\fctmongeproductismonge*

\label{sec:appendix-fctmongeproductismonge}

\begin{proof}
    Let $A$ be of size $p \times q$ and $B$ be of size $q \times r$.
    Fix some $i \in \fragmentco{0}{p-1}$ and $k \in \fragmentco{0}{r-1}$.
    We prove $C_{i, k} + C_{i + 1, k + 1} \le C_{i, k + 1} + C_{i + 1, k}$.
    Pick $j = \wit{A}{B}_{i + 1, k}$ and $j' = \wit{A}{B}_{i, k + 1}$ so that $C_{i + 1, k} = A_{i + 1, j} + B_{j, k}$ and $C_{i, k + 1} = A_{i, j'} + B_{j', k + 1}$.
    Suppose $j \le j'$.
    We have
    \begin{align*}
        C_{i, k} + C_{i + 1, k + 1} &\le (A_{i, j} + B_{j, k}) + (A_{i + 1, j'} + B_{j', k + 1})\\
                                    &= (A_{i, j} + A_{i + 1, j'}) + (B_{j, k} + B_{j', k + 1})\\
                                    &\le (A_{i, j'} + A_{i + 1, j}) + (B_{j, k} + B_{j', k + 1})&\text{\hfill by \cref{cor:non-local-monge-property}}\\
                                    &= (A_{i, j'} + B_{j', k + 1}) + (A_{i + 1, j} + B_{j, k})\\
                                    &= C_{i, k + 1} + C_{i + 1, k}.
    \end{align*}
    The case of $j > j'$ is analogous.
\end{proof}

\lmcompression*

\label{sec:appendix-lmcompression}

\begin{proof}
    We first describe how to convert $A^*$ into a $p \times q^*$ Monge matrix $A'$ such that $p \le \delta(A^*) + 1$ and $\delta(A') = \delta(A^*)$.

    Call the $i$-th row of $A^*$ for $i \in \fragmentoo{0}{p^*}$ \emph{redundant} if $\deltai{i - 1}(A^*) = 0$.
    To obtain $A'$, we remove all redundant rows of $A^*$.
    The matrix $A'$ is Monge by \cref{lm:monge-submatrix}.

    We now show how to compute the condensed representation of $A'$ in $\Oh(p^* + q^* + \delta(A^*))$ time.  
    To find redundant rows, we iterate over $\core(A^*)$ and mark all rows that contain any core elements.
    The topmost row of $A'$ is the same as the topmost row of $A^*$.
    To get the leftmost column of $A'$, we filter out all elements of the leftmost column of $A^*$ that lie in the redundant rows.
    It remains to compute the core of $A'$.
    Consider a function $f$ that maps the rows of $A'$ to the corresponding rows of~$A^*$.
    For any $i \in \fragmentco{0}{p - 1}$ and $j \in \fragmentco{0}{q-1}$, due to \cref{lm:calculate-value-from-core}, we have
    \begin{align*}
        \dens{A'}_{i, j} &= A'_{i, j + 1} + A'_{i + 1, j} - A'_{i, j} - A'_{i + 1, j + 1}\\
                         &= A^*_{f(i), j + 1} + A^*_{f(i + 1), j} - A^*_{f(i), j} - A^*_{f(i + 1), j + 1}\\
                         &= \deltasigma(A^*\fragmentcc{f(i)}{f(i+1)}\fragmentcc{j}{j+1})\\
                         &= A^{*\square}_{f(i), j} + A^{*\square}_{f(i) + 1, j} + \cdots + A^{*\square}_{f(i + 1) - 1, j}\\
                         &= A^{*\square}_{f(i + 1) - 1, j},
    \end{align*}
    where the last equality holds as $A^{*\square}_{f(i), j} = A^{*\square}_{f(i) + 1, j} = \cdots = A^{*\square}_{f(i + 1) - 2, j} = 0$ because $f(i)$ and $f(i+1)$ are two consecutive non-redundant rows of $A^*$.
    Therefore, the set $\core(A') = \{(i, j, \dens{A'}_{i, j}) \mid i \in \fragmentco{0}{p - 1}, j \in \fragmentco{0}{q-1}, \dens{A'}_{i, j} \neq 0\} = \{(f^{-1}(i + 1) - 1, j, A^{*\square}_{i, j}) \mid (i, j, A^{*\square}_{i, j}) \in \core(A^*)\}$ is of size $\delta(A') = \delta(A^*)$ and can be computed in time $\Oh(\delta(A^*) + 1)$.
    Furthermore, for every $i \in \fragmentco{0}{p - 1}$, we have $\deltai{i}(A') \neq 0$ by the definition of $A'$, and thus $p \le \delta(A') + 1 = \delta(A^*) + 1$.

    Analogously (up to a transposition), in $\Oh(q^*+r^*+\delta(B^*))$ time we convert $B^*$ into a $q^* \times r$ Monge matrix $B'$ such that $r \le \delta(B^*) + 1$ and $\delta(B') = \delta(B^*)$.
    
    \newcommand{\jjj}{j_{\textup{opt}}}
    The last step of the compression algorithm is to convert $A'$ and $B'$ into Monge matrices $A$ and $B$ of sizes $p \times q$ and $q \times r$, respectively, such that $q \le \delta(A')+\delta(B')+1$, $A \minplus B = A' \minplus B'$, $\delta(A) = \delta(A')$, and $\delta(B) = \delta(B')$.
    Consider an arbitrary segment $\fragmentco{j_L}{j_R} \subseteq \fragmentco{0}{q^*}$ such that $\delta(A'\fragmentco{0}{p}\fragmentco{j_L}{j_R}) = \delta(B'\fragmentco{j_L}{j_R}\fragmentco{0}{r}) = 0$.
    Pick $\jjj \coloneqq \argmin \{ A'_{0, j} + B'_{j, 0} \mid j \in \fragmentco{j_L}{j_R} \}$.
    For all $i \in \fragmentco{0}{p}$, $j \in \fragmentco{j_L}{j_R}$, and $r \in \fragmentco{0}{r}$, we have
    \begin{align*}
        A'_{i, j} + B'_{j, k} &= (A'_{0, j} + A'_{i, \jjj} - A'_{0, \jjj}) + (B'_{j, 0} + B'_{\jjj, k} - B'_{\jjj, 0})\\
                              &= A'_{i, \jjj} + B'_{\jjj, k} + (A'_{0, j} + B'_{j, 0}) - (A'_{0, \jjj} + B'_{\jjj, 0})\\
                              &\ge A'_{i, \jjj} + B'_{\jjj, k}
    \end{align*}
    by \cref{lm:calculate-value-from-core} and the definition of $\jjj$.
    Therefore, indices $j \in \fragmentco{j_L}{j_R} \setminus \{\jjj\}$ are not needed for the computation of any entry $(A' \minplus B')_{i,k}$.
    We can thus split $\fragmentco{0}{q^*}$ into $q \le \delta(A') + \delta(B') + 1$ such disjoint segments $\fragmentco{j_L}{j_R}$ and pick a single ``optimal'' index $\jjj$ in each such segment.
    To obtain the resulting matrices $A$ and $B$, we drop all columns of $A'$ and rows of $B'$ that correspond to indices that are not picked.
    The condensed representations of $A$ and $B$ are obtained from the condensed representations of $A'$ and $B'$ similarly to the conversion of $A^*$ into $A'$.
    Note that $A$ and $B$ are Monge matrices due to \cref{lm:monge-submatrix}.
    The whole compression procedure takes $\Oh(p^* + q^* + r^* + \delta(A^*) + \delta(B^*))$ time.

    \medskip

    We now describe how to obtain the condensed representation of $C^* \coloneqq A^* \minplus B'$ from the condensed representation of $C \coloneqq A \minplus B = A' \minplus B'$.
    This constitutes one half of the decompression algorithm.
    Note that $C_{i, k} = C^*_{f(i), k}$ for every $i \in \fragmentco{0}{p}$ and $k \in \fragmentco{0}{r}$.
    In particular, the topmost row of $C^*$ coincides with the topmost row of $C$.

    Consider some $i \in \fragmentco{0}{p^*-1}$ with $\deltai{i}(A^*) = 0$.
    For every $j \in \fragmentco{0}{q}$,  we have
    \[
        A^*_{i + 1, j} = A^*_{i, j} + (A^*_{i + 1, 0} - A^*_{i, 0}) - \deltasigma(A^*\fragmentcc{i}{i+1}\fragmentcc{0}{j}) = A^*_{i, j} + (A^*_{i + 1, 0} - A^*_{i, 0})
    \]
    due to \cref{lm:calculate-value-from-core}.
    Therefore, for every $k \in \fragmentco{0}{r}$, we have
    \begin{align*}
        C^*_{i + 1, k} &= \min \{ A^*_{i + 1, j} + B'_{j, k} \mid j \in \fragmentco{0}{q} \}\\
                       &= \min \{ A^*_{i, j} + (A^*_{i + 1, 0} - A^*_{i, 0}) + B'_{j, k} \mid j \in \fragmentco{0}{q} \}\\
                     &= \min \{ A^*_{i, j} + B'_{j, k} \mid j \in \fragmentco{0}{q} \} + (A^*_{i + 1, 0} - A^*_{i, 0})\\
                     &= C^*_{i, k} + (A^*_{i + 1, 0} - A^*_{i, 0}).
    \end{align*}
    In particular, we may use this equality to compute the values in the leftmost column of $C^*$: for $i \in \fragmentco{0}{p}$, we have $C^*_{f(i), 0} = C_{i, 0}$ and $C^*_{i', 0} = C^*_{f(i), 0} + (A^*_{i', 0} - A^*_{f(i), 0})$ for $i' \in \fragmentoo{f(i)}{f(i + 1)}$.
    Furthermore, for $i' \in \fragmentco{f(i)}{f(i+1)-1}$, we have $\deltai{i'}(C^*) = 0$.
    Finally, 
    \begin{align*}
        C^{*\square}_{f(i + 1) - 1, k} &= C^{*\square}_{f(i), k} + C^{*\square}_{f(i) + 1, k} + \cdots + C^{*\square}_{f(i + 1) - 1, k}\\
                                       &= \deltasigma(C^*\fragmentcc{f(i)}{f(i+1)}\fragmentcc{k}{k+1})\\
                                       &= C^*_{f(i), k + 1} + C^*_{f(i + 1), k} - C^*_{f(i), k} - C^*_{f(i + 1), k + 1}\\
                                       &= C_{i, k + 1} + C_{i + 1, k} - C_{i, k} - C_{i + 1, k + 1}\\
                                       &= \dens{C}_{i, k}
    \end{align*}
    for every $k \in \fragmentco{0}{r - 1}$ due to the definition of $C$ and \cref{lm:calculate-value-from-core}.
    Therefore, the set $\core(C^*) = \{(i, k, C^{*\square}_{i, k}) \mid i \in \fragmentco{0}{p^*-1}, k \in \fragmentco{0}{r - 1}, C^{*\square}_{i, k} \neq 0\} = \{ (f(i+1)-1, k, \dens{C}_{i, k}) \mid (i, k, \dens{C}_{i, k}) \in \core(C)\}$ can be computed in time $\Oh(\delta(C)) = \Oh(\delta(A) + \delta(B)) = \Oh(\delta(A^*) + \delta(B^*))$, where $\delta(C) \le 2 \cdot (\delta(A) + \delta(B))$ due to \cref{lm:core-preservation}.
    The whole procedure takes $\Oh(p^* + r + \delta(A^*) + \delta(B^*))$ time.

    We showed how to decompress $A \minplus B=A'\minplus B'$ into $A^* \minplus B'$.
    The decompression of $A^* \minplus B'$ into $A^* \minplus B^*$ can be done symmetrically in time $\Oh(p^* + r^* + \delta(A^*) + \delta(B^*))$.
\end{proof}

\lmcorebasedmultiplicationalgorithm*

\label{sec:appendix-lmcorebasedmultiplicationalgorithm}

\begin{proof}
    We now formally describe the recursive procedure $\recmult$ that is run for a $p^* \times q^*$ matrix $A^*$ and a $q^* \times r^*$ matrix $B^*$.
    Consult \cref{alg:core-based-multiplication-algorithm} for the pseudocode.

    We first apply the \compress algorithm of \cref{lm:compression} to obtain the condensed representations of a Monge matrix $A$ of size $p \times q$ with $p \le \delta(A^*) + 1 = \delta(A) + 1$ and $q \le \delta(A^*) + \delta(B^*) + 1$ and a Monge matrix $B$ of size $q \times r$ with $r \le \delta(B^*) + 1 = \delta(B) + 1$.
    We then compute the condensed representation of $C \coloneqq A \minplus B$, and finally we use the \decompress algorithm of \cref{lm:compression} to obtain the condensed representation of $A^* \minplus B^*$.

    If $p = 1$ and $r = 1$, we compute the only entry of the $1 \times 1$ matrix $C$ (and thus, its condensed representation) trivially.
    Otherwise, we pick the largest index $m \in \fragmentoo{0}{q}$ with $\delta(A\fragmentco{0}{p}\fragmentco{0}{m}) + \delta(B\fragmentco{0}{m}\fragmentco{0}{r}) \le (\delta(A) + \delta(B)) / 2$.
    Such an $m$ can be computed by traversing the lists $\core(A)$ and $\core(B)$ sorted (using counting sort) according to the second and the first coordinate respectively.

    Let $A^L \coloneqq A\fragmentco{0}{p}\fragmentco{0}{m}$, $B^L \coloneqq B\fragmentco{0}{m}\fragmentco{0}{r}$, $A^R \coloneqq A\fragmentco{0}{p}\fragmentco{m}{q}$, and $B^R \coloneqq B\fragmentco{m}{q}\fragmentco{0}{r}$.
    We compute the condensed representations of all these matrices using \cref{lm:get-submatrix}.
    Furthermore, we recursively compute the condensed representations of $C^L \coloneqq A^L \minplus B^L$ and $C^R \coloneqq A^R \minplus B^R$, and build $\lco(C^L)$ and $\lco(C^R)$ of \cref{lm:local-core-oracle}.
    As shown below, the desired matrix $C$ is the element-wise minimum of $C^L$ and $C^R$.
    Moreover, there exists a non-increasing sequence $(lst_i)_{i \in \fragmentco{0}{p}}$ such that $C_{i, j} = C^L_{i, j} \le C^R_{i, j}$ for $j \in \fragmentcc{0}{lst_i}$ and $C_{i, j} = C^R_{i, j} < C^L_{i, j}$ for $j \in \fragmentoo{lst_i}{r}$.
    See \cref{fig:monge-multuplication-conquer-step}, where the blue staircase represents the sequence $lst$.
    We now compute the sequence $lst$ along with the core elements of $C$ located on the ``border'' between the values of $C$ taken from $C^L$ and the values of $C$ taken from $C^R$.
    We start traversing the matrix $C$ in the direction from bottom-left to top-right.
    At every point in time, we consider some cell $(i, j)$ with $i \in \fragmentco{0}{p}$ and $j \in \fragmentco{-1}{q}$.
    We maintain all values in $C^L\fragmentcc{i}{i+1}\fragmentcc{j}{j+1}$ and $C^R\fragmentcc{i}{i+1}\fragmentcc{j}{j+1}$ that are well-defined (that is, within the matrix boundaries).
    Initially, we set $i \coloneqq p - 1$ and $j \coloneqq -1$.
    At each step, we compute the values in $C\fragmentcc{i}{i+1}\fragmentcc{j}{j+1}$ as the element-wise minima of the values in $C^L\fragmentcc{i}{i+1}\fragmentcc{j}{j+1}$ and $C^R\fragmentcc{i}{i+1}\fragmentcc{j}{j+1}$, and if $\dens{C}_{i, j} \neq 0$, we add $(i, j, \dens{C}_{i, j})$ to the core of~$C$.
    Furthermore, if $j + 1 < r$ and $C^L_{i, j + 1} \le C^R_{i, j + 1}$, we have that $lst_i \ge j + 1$, and thus we increment $j$ by one and compute the values in $C^L\fragmentcc{i}{i+1}\fragmentcc{j+1}{j+1}$ and $C^R\fragmentcc{i}{i+1}\fragmentcc{j+1}{j+1}$ using the horizontally adjacent recomputation of $\lco(C^L)$ and $\lco(C^R)$.
    Otherwise, we set $lst_i = j$ and decrement $i$ by one.
    We maintain the values in $C^L\fragmentcc{i}{i}\fragmentcc{j}{j+1}$ and $C^L\fragmentcc{i}{i}\fragmentcc{j}{j+1}$ using vertically adjacent recomputation of $\lco(C^L)$ and $\lco(C^R)$.
    We terminate the process when $i$ becomes negative.
    
    After computing $lst_i$ for all $i \in \fragmentco{0}{p}$, we traverse all the elements $(i, j, C^{L\square}_{i, j}) \in \core(C^L)$ and add them to $\core(C)$ if $j < lst_{i + 1}$.
    Analogously, we traverse all the elements $(i, j, C^{R\square}_{i, j}) \in \core(C^R)$ and add them to $\core(C)$ if $j > lst_i$.
    Finally, we compute the elements in the topmost row and the leftmost column of $C$ as the element-wise minima of the elements in the topmost rows and the leftmost columns of $C^L$ and $C^R$.
    Having computed $\core(C)$, the topmost row of $C$, and the leftmost column of $C$, we can return the condensed representation of $C$.

    \medskip
    
    We now prove the correctness of the presented algorithm.
    We first show that $C$ is the element-wise minimum of $C^L$ and $C^R$.
    By definition, for every $i \in \fragmentco{0}{p}$ and $k \in \fragmentco{0}{r}$, we have
    \begin{align*}
        \min\{C^L_{i, k}, C^R_{i, k}\} &= \min\{\min \{A_{i, j} + B_{j, k} \mid j \in \fragmentco{0}{m}\}, \min \{A_{i, j} + B_{j, k} \mid j \in \fragmentco{m}{q}\}\}\\
                                       &= \min\{A_{i, j} + B_{j, k} \mid j \in \fragmentco{0}{q}\}\\
                                       &= C_{i, k}.
\end{align*}
    Furthermore, $C_{i, k} = C^L_{i, k} \le C^R_{i, k}$ if $\wit{A}{B}_{i, k} < m$ and $C_{i, k} = C^R_{i, k} < C^L_{i, k}$ if $\wit{A}{B}_{i, k} \ge m$.
    As $\wit{A}{B}$ is non-decreasing by rows and columns due to \cref{lm:witness-monotonicity}, there indeed exists a non-decreasing sequence $(lst_i)_{i\in  \fragmentco{0}{p}}$ with values in $\fragmentco{-1}{r}$ such that $C_{i, k} = C^L_{i, k} \le C^R_{i, k}$ for $k \in \fragmentcc{0}{lst_i}$ and $C_{i, k} = C^R_{i, k} < C^L_{i, k}$ for $k \in \fragmentoo{lst_i}{r}$.
    In short, $lst_i=\max(\{-1\}\cup \{j\in \fragmentco{0}{q} \mid C^L_{i, j} \le C^R_{i, j}\})$.
    We now prove the correctness of the algorithm that computes $lst$.
    In the loop, we start from $(i, j) = (p - 1, -1)$ and increment $j$ while $j < lst_i$ (that is, while $C^L_{i, j + 1} \le C^R_{i, j + 1}$).
    We always maintain the invariant that $j \le lst_i$.
    When we have $j = r - 1$ or $C^L_{i, j + 1} > C^R_{i, j + 1}$, we assign $lst_i = j$, which is indeed correct.
    As $lst$ is non-decreasing, we have $lst_{i - 1} \ge lst_i = j$, so the invariant holds when we decrement $i$.

    We now prove that we computed the core of $C$ correctly.
    At each step, we compute $\dens{C}_{i, j}$ and add $(i, j, \dens{C}_{i, j})$ to the core of $C$ if and only if $\dens{C}_{i, j} \neq 0$.
    Note that in the $i$-th row for $i \in \fragmentco{0}{p-1}$, we visited exactly the cells $(i, j)$ with $j \in \fragmentcc{lst_{i+1}}{lst_i}$.
    Therefore, we computed all core elements $(i, j, \dens{C}_{i, j}) \in \core(C)$ with $j \in \fragmentcc{\max\{0, lst_{i+1}\}}{\min\{r - 2, lst_i\}}$.
    If this is not the case, we either have $j \in \fragmentco{0}{lst_{i+1}}$ or $j \in \fragmentoo{lst_i}{r - 1}$.
    In the first case, we have $j + 1 \le lst_{i+1} \le lst_i$.
    Therefore, $C\fragmentcc{i}{i+1}\fragmentcc{j}{j+1} = C^L\fragmentcc{i}{i+1}\fragmentcc{j}{j+1}$.
    Hence, $\{ (i, j, \dens{C}_{i, j}) \in \core(C) \mid j < lst_{i+1} \} = \{(i, j, C^{L\square}_{i, j}) \in \core(C^L) \mid j < lst_{i + 1}\}$.
    Analogously, in the second case, we have $j > lst_i \ge lst_{i+1}$, and thus $C\fragmentcc{i}{i+1}\fragmentcc{j}{j+1} = C^R\fragmentcc{i}{i+1}\fragmentcc{j}{j+1}$.
    Hence, $\{ (i, j, \dens{C}_{i, j}) \in \core(C) \mid j > lst_i \} = \{(i, j, C^{R\square}_{i, j}) \in \core(C^R) \mid j > lst_i\}$.
    Therefore, the core of $C$ and the entire condensed representation of $C$ were computed correctly.

    \medskip

    It remains to analyze the time complexity of the presented recursive procedure.
    We first analyze the time complexity of a single call without further recursive calls.
    The compression and decompression procedures take time $\Oh(p^* + q^* + r^* + \delta(A^*) + \delta(B^*))$.
    If $p = r = 1$, we compute $C$ in $\Oh(q)$ time.
    Otherwise, we sort $\core(A)$ and $\core(B)$ using counting sort and pick $m$ in $\Oh(\delta(A) + \delta(B) + q) = \Oh(\delta(A) + \delta(B))$ time.
    Building the condensed representations of $A^L, A^R, B^L$, and $B^R$ using \cref{lm:get-submatrix} takes $\Oh(p + q + r + \delta(A) + \delta(B)) = \Oh(\delta(A) + \delta(B))$ time.
    Note that $A^L$ and $A^R$ are disjoint contiguous submatrices of $A$, and thus $\delta(A^L) + \delta(A^R) \le \delta(A)$.
    Analogously, $\delta(B^L) + \delta(B^R) \le \delta(B)$.
    In particular, we have $\delta(A^L), \delta(A^R) \le \delta(A)$ and $\delta(B^L), \delta(B^R) \le \delta(B)$.
    As $C^L = A^L \minplus B^L$ and $C^R = A^R \minplus B^R$, we have $\delta(C^L), \delta(C^R) \le 2 \cdot (\delta(A) + \delta(B))$ due to \cref{lm:core-preservation}.
    Therefore, constructing $\lco(C^L)$ and $\lco(C^R)$ takes time $\Oh(\delta(A) + \delta(B))$.
    After that, we initialize $i = p - 1$ and $j = -1$ and perform at most $p + r$ loop iterations as each iteration either decrements $i$ or increments $j$.
    Therefore, all loop operations except for maintaining the values of $C^L\fragmentcc{i}{i+1}\fragmentcc{j}{j+1}$ and $C^R\fragmentcc{i}{i+1}\fragmentcc{j}{j+1}$ take $\Oh(p + r)$ time.

    We now analyze the time of maintaining the values of $C^L\fragmentcc{i}{i+1}\fragmentcc{j}{j+1}$.
    We initially set $i = p - 1$ and $j = -1$ and only compute $C^L_{i, j + 1}$ in constant time as it is the only existing entry of $C^L\fragmentcc{i}{i+1}\fragmentcc{j}{j+1}$.
    Furthermore, whenever in the loop $j$ is incremented, the values $C^L_{i, j}$ and $C^L_{i + 1, j}$ for the new $j$ are already computed, and computing $C^L_{i, j + 1}$ and $C^L_{i + 1, j + 1}$ using horizontally adjacent recomputation takes $\Oh(\deltaj{j}(C^L) + 1)$ time.
    Analogously, when $i$ is decremented, it takes $\Oh(\deltai{i}(C^L) + 1)$ time to recompute the values using the vertically adjacent recomputation.
    In total, it takes $\Oh(1 + \sum_{j \in \fragmentco{0}{r-1}} (\deltaj{j}(C^L) + 1) + \sum_{i \in \fragmentco{0}{p-1}} ({\deltai{i}(C^L) + 1})) = \Oh(\delta(C^L) + p + r)$ time to maintain the values of $C^L\fragmentcc{i}{i+1}\fragmentcc{j}{j+1}$.
    Analogously, maintaining the values of $C^R\fragmentcc{i}{i+1}\fragmentcc{j}{j+1}$ takes $\Oh(\delta(C^R) + p + r)$ time.
    Overall, it takes $\Oh(\delta(C^L) + \delta(C^R) + p + r) = \Oh(\delta(A) + \delta(B))$ time to traverse the matrix $C$.

    Finally, filtering out the elements of $\core(C^L)$ and $\core(C^R)$ and computing the values of the topmost row and the leftmost column of $C$ takes $\Oh(\delta(C^L) + \delta(C^R) + p + r) = \Oh(\delta(A) + \delta(B))$ time.
    Hence, the complete call without recursive calls takes time $\Oh(p^* + q^* + r^* + \delta(A^*) + \delta(B^*))$.
    
    We now prove that, including the recursive calls, the algorithm takes $\Oh(p^* + q^* + r^* + (\delta(A^*) + \delta(B^*)) \log (1 + \delta(A^*) + \delta(B^*)))$ time.
    Note that a call with parameters $(p^*_i, q^*_i, r^*_i, \delta(A^*_i), \delta(B^*_i))$ makes two recursive calls with parameters $(p_i, m_i, r_i, \delta(A^L_i), \delta(B^L_i))$ and $(p_i, q_i - m_i, r_i, \delta(A^R_i), \delta(B^R_i))$, where $p_i \le \delta(A^*_i) + 1 \le 2 \cdot (\delta(A^*_i) + \delta(B^*_i))$,\footnote{This inequality holds because further recursive calls are made only if $\max\{p_i, r_i\} > 1$, and thus $\max\{\delta(A^*_i), \delta(B^*_i)\} > 0$.} $q_i \le \delta(A^*_i) + \delta(B^*_i) + 1 \le 2 \cdot (\delta(A^*_i) + \delta(B^*_i))$, $r_i \le \delta(B^*_i) + 1 \le 2 \cdot (\delta(A^*_i) + \delta(B^*_i))$, $\delta(A^L_i) + \delta(A^R_i) \le \delta(A_i^*)$, and $\delta(B^L_i) + \delta(B^R_i) \le \delta(B_i^*)$.
    Therefore, at each level of the recursion except for the first one, the sum of $\delta(A^*_i)$ is bounded by the initial $\delta(A^*)$, the sum of $\delta(B^*_i)$ is bounded by the initial $\delta(B^*)$, the sum of $p^*_i$ is bounded by $4 \cdot (\delta(A^*) + \delta(B^*))$, the sum of $q^*_i$ is bounded by $2 \cdot (\delta(A^*) + \delta(B^*))$, and the sum of $r^*_i$ is bounded by $4 \cdot (\delta(A^*) + \delta(B^*))$.
    Hence, at each level of the recursion except for the first one, the algorithm works in total time $\Oh(\delta(A^*) + \delta(B^*))$.
    At the first level of the recursion, the algorithm takes $\Oh(p^* + q^* + r^* + \delta(A^*) + \delta(B^*))$ time.
    
    Therefore, it remains to show that there are at most $1 + \log(1 + \delta(A^*) + \delta(B^*))$ levels of recursion in total.
    We show that $\delta(A^L_i) + \delta(B^L_i) \le (\delta(A_i) + \delta(B_i)) / 2$ and $\delta(A^R_i) + \delta(B^R_i) \le (\delta(A_i) + \delta(B_i)) / 2$.
    The inequality $\delta(A^L_i) + \delta(B^L_i) \le (\delta(A_i) + \delta(B_i)) / 2$ follows from the definition of $m_i$.
    On the other hand, as $m \in \fragmentoc{0}{q}$ is the largest index with this property, we have that $\delta(A_i\fragmentco{0}{p}\fragmentcc{0}{m}) + \delta(B_i\fragmentcc{0}{m}\fragmentco{0}{r}) > (\delta(A_i) + \delta(B_i)) / 2$, and thus $\delta(A^R_i) + \delta(B^R_i) = (\delta(A_i) - \delta(A_i\fragmentco{0}{p}\fragmentcc{0}{m})) + (\delta(B_i) - \delta(B_i\fragmentcc{0}{m}\fragmentco{0}{r})) \le (\delta(A_i) + \delta(B_i)) / 2$.
    Thus, after at most $\log(1 + \delta(A^*) + \delta(B^*))$ iterations of the recursion, we have $\delta(A^*_i) + \delta(B^*_i) = 0$; such a recursive call is a leaf recursive call since $p_i \le \delta(A^*_i) + 1 = 1$ and $r_i \le \delta(B^*_i) + 1 = 1$.
\end{proof}

\lmproductwitnessreconstruction*

\label{sec:appendix-lmproductwitnessreconstruction}

\begin{proof}

    We slightly modify the algorithm of \cref{lm:core-based-multiplication-algorithm}.
    In addition to the returned condensed representation of $C^*$, we also store the value $m$ and the computed list $lst$ at each non-leaf level of the recursion, and $\wit{A}{B}_{0, 0}$ in the leaf case of $p=r=1$.
    Furthermore, in the compression-decompression algorithm of \cref{lm:compression}, we store, for all rows and columns of $A$ and $B$ to which rows and columns of $A^*$ and $B^*$ they correspond.
    In particular, for every $k \in \fragmentco{0}{q}$, let $h(k)$ be the column of $A^*$ that corresponds to the $k$-th column of $A$.
    Note that if the $i$-th row of $A^*$ is redundant, then $\wit{A^*}{B^*}\fragmentcc{i}{i}\fragmentco{0}{r^*} = \wit{A^*}{B^*}\fragmentcc{i - 1}{i-1}\fragmentco{0}{r^*}$, and thus, if row $i$ of $C^*$ corresponds to row $i'$ of $C$ and column $j$ of $C^*$ corresponds to column $j'$ of $C$, we have $\wit{A^*}{B^*}_{i, j} = h(\wit{A}{B}_{i', j'})$.

    To compute $\wit{A^*}{B^*}_{i, j}$, we start at the root of the recursion.
    In constant time, we replace $(i,j)$ with $(i',j')$ such that $\wit{A^*}{B^*}_{i, j} = h(\wit{A}{B}_{i', j'})$.
    After that, we know that $\wit{A}{B}_{i', j'} < m$ if and only if $j' \le lst_{i'}$.
    This can be decided in constant time.
    If $j' \le lst_{i'}$, we have $\wit{A}{B}_{i', j'} = \wit{A^L}{B^L}_{i', j'}$, which can be computed recursively.
    Otherwise, if $j' > lst_{i'}$, we have $\wit{A}{B}_{i', j'} = \wit{A^R}{B^R}_{i', j'} + m$, and $\wit{A^R}{B^R}_{i', j'}$ can be computed recursively.
    In the leaf level of the recursion, the smallest witness of the only entry of $C$ is stored explicitly.

    As the depth of the recursion is at most $1 + \log(1 + \delta(A^*) + \delta(B^*))$, and we spend constant time at each level of the recursion, the total time complexity for the data structure query is $\Oh(\log(2 + \delta(A^*) + \delta(B^*)))$.

    \medskip

    It remains to show how to store such a data structure in linear space.
    For each recursive call, we store some constant number of integers and some lists.
    The lists we store are the $lst$ list and the correspondence between rows and columns of $A$ and $A^*$ and $B$ and $B^*$.
    All these lists are monotone, of length at most $p^* + q^* + r^*$, and with values in $\fragmentcc{-1}{p^* + q^* + r^*}$, so they can be stored as bit-masks occupying $\Oh(1 + (p^* + q^* + r^*) / w)$ space in total, where $w$ is the machine word size.
    Using the rank-select data structure \cite{Jacobson1989, Munro1996, Clark1996, BGKS2015, MNV2016}, we provide constant-time access to lists, so the query time complexity does not change.
    Including this data structure affects neither the asymptotic space complexity nor the construction time.

    It remains to analyze the space complexity of the whole data structure.
    The same way as in the analysis of the algorithm of \cref{lm:core-based-multiplication-algorithm}, the sum of $p^*_i + q^*_i + r^*_i$ over all recursive calls is at most $\Oh(p^* + q^* + r^* + (\delta(A^*) + \delta(B^*)) \cdot \log(1 + \delta(A^*) + \delta(B^*)))$.
    Since ${\log(1 + \delta(A^*) + \delta(B^*))}=\Oh(w)$, we get that the sum of $\Oh((p^*_i + q^*_i + r^*_i) / w)$ over all recursive calls is limited by $\Oh(p^* + q^* + r^* + \delta(A^*) + \delta(B^*))$.
    Furthermore, since there are at most $1 + \log(1 + \delta(A^*) + \delta(B^*))$ recursion levels, there are $\Oh(1 + \delta(A^*) + \delta(B^*))$ recursive calls in total, and thus the sum of $\Oh(1 + (p^*_i + q^*_i + r^*_i) / w)$ over all recursive calls is bounded by $\Oh(p^* + q^* + r^* + \delta(A^*) + \delta(B^*))$.
\end{proof}

\section{Range LIS Queries: Formal Proofs}\label{app:range-lis}

In this section we give all necessary definitions and formal proofs of the algorithms described in \cref{sec:range-lis}.
We start by defining the problems at hand.

\begin{definition}
    The \emph{Longest Increasing Subsequence (LIS)} problem, given a sequence $s\fragmentco{0}{n}$, asks to find the maximum length of a (not necessarily contiguous) subsequence of $s$ with strictly increasing values.
    The \emph{reporting} version further asks to output the underlying subsequence itself.

    The problem of \emph{Range LIS Value Queries} asks to preprocess a given sequence $s$ and then answer queries for LIS length on contiguous subsegments of $s$.
    The \emph{Range LIS Reporting Queries} ask to solve the reporting version of the LIS problem on contiguous subsegments of $s$.
%\lipicsEnd
\end{definition}

We now define the notion of a subpermutation that allows some elements of the sequence $s$ to be inaccessible.

\begin{definition}
    A sequence $(s\position{i})_{i \in \fragmentco{0}{n}}$ is called a \emph{subpermutation} if $s\position{i} \in \fragmentco{0}{n} \cup \{\plh\}$ for all $i \in \fragmentco{0}{n}$, and the subsequence of $s$ consisting of all elements that are not equal to~$\plh$ is a permutation.
    Such a subsequence is called the \emph{underlying permutation} of $s$ and is denoted by $s^*$.

    In other words, a subpermutation is a permutation in which the \emph{placeholder} symbol $\plh$ is inserted at some positions.
%\lipicsEnd
\end{definition}

\begin{definition}
    For a subpermutation $s$ of length $n$ and indices $i, j \in \fragmentcc{0}{n}$ with $i < j$, let $\lis(s\fragmentco{i}{j})$ be the length of the longest increasing subsequence of $s\fragmentco{i}{j}$ that does not contain any placeholder symbols $\plh$.
%\lipicsEnd
\end{definition}

\begin{definition}\label{def:alignment-graph}
    For a subpermutation $s$ of length $n$, we define the following edge-weighted directed graph called the \emph{alignment graph} of $s$ and denoted by $G^s$ (see \cref{fig:alignment-graph} for an example).
    The vertex set of $G^s$ is $\fragmentcc{0}{n} \times \fragmentcc{0}{|s^*|}$.
    Furthermore, $G^s$ contains the following edges:
    \begin{itemize}
        \item $(x, y) \to (x, y + 1)$ of weight $0$ for all $x \in \fragmentcc{0}{n}$ and $y \in \fragmentco{0}{|s^*|}$;
        \item $(x, y) \to (x + 1, y)$ of weight $0$ for all $x \in \fragmentco{0}{n}$ and $y \in \fragmentcc{0}{|s^*|}$;
        \item $(i, s\position{i}) \to (i + 1, s\position{i} + 1)$ of weight $1$ for all $i \in \fragmentco{0}{n}$ with $s\position{i} \neq \plh$;
        \item $(x, y) \to (x - 1, y)$ of weight $-2$ for all $x \in \fragmentoc{0}{n}$ and $y \in \fragmentcc{0}{|s^*|}$.
%\lipicsEnd
    \end{itemize}

    \begin{figure}
        \begin{center}
            \begin{tikzpicture}
                \node[inner sep = 1pt,circle,fill=black] (n0_0) at (0,0) {};
\node[inner sep = 1pt,circle,fill=black] (n0_1) at (0,1) {};
\draw[-latex,darkred] (n0_0) -- (n0_1);
\node[inner sep = 1pt,circle,fill=black] (n0_2) at (0,2) {};
\draw[-latex,darkred] (n0_1) -- (n0_2);
\node[inner sep = 1pt,circle,fill=black] (n0_3) at (0,3) {};
\draw[-latex,darkred] (n0_2) -- (n0_3);
\node[inner sep = 1pt,circle,fill=black] (n0_4) at (0,4) {};
\draw[-latex,darkred] (n0_3) -- (n0_4);
\node[inner sep = 1pt,circle,fill=black] (n0_5) at (0,5) {};
\draw[-latex,darkred] (n0_4) -- (n0_5);
\node[inner sep = 1pt,circle,fill=black] (n1_0) at (1,0) {};
\draw[-latex,darkred] (n0_0) to[out=45,in=135] (n1_0);
\draw[-latex,darkblue] (n1_0) to[out=225,in=315] (n0_0);
\node[inner sep = 1pt,circle,fill=black] (n1_1) at (1,1) {};
\draw[-latex,darkred] (n0_1) to[out=45,in=135] (n1_1);
\draw[-latex,darkblue] (n1_1) to[out=225,in=315] (n0_1);
\draw[-latex,darkred] (n1_0) -- (n1_1);
\node[inner sep = 1pt,circle,fill=black] (n1_2) at (1,2) {};
\draw[-latex,darkred] (n0_2) to[out=45,in=135] (n1_2);
\draw[-latex,darkblue] (n1_2) to[out=225,in=315] (n0_2);
\draw[-latex,darkred] (n1_1) -- (n1_2);
\draw[-latex,very thick,darkgreen] (n0_1) -- (n1_2);
\node[inner sep = 1pt,circle,fill=black] (n1_3) at (1,3) {};
\draw[-latex,darkred] (n0_3) to[out=45,in=135] (n1_3);
\draw[-latex,darkblue] (n1_3) to[out=225,in=315] (n0_3);
\draw[-latex,darkred] (n1_2) -- (n1_3);
\node[inner sep = 1pt,circle,fill=black] (n1_4) at (1,4) {};
\draw[-latex,darkred] (n0_4) to[out=45,in=135] (n1_4);
\draw[-latex,darkblue] (n1_4) to[out=225,in=315] (n0_4);
\draw[-latex,darkred] (n1_3) -- (n1_4);
\node[inner sep = 1pt,circle,fill=black] (n1_5) at (1,5) {};
\draw[-latex,darkred] (n0_5) to[out=45,in=135] (n1_5);
\draw[-latex,darkblue] (n1_5) to[out=225,in=315] (n0_5);
\draw[-latex,darkred] (n1_4) -- (n1_5);
\node[inner sep = 1pt,circle,fill=black] (n2_0) at (2,0) {};
\draw[-latex,darkred] (n1_0) to[out=45,in=135] (n2_0);
\draw[-latex,darkblue] (n2_0) to[out=225,in=315] (n1_0);
\node[inner sep = 1pt,circle,fill=black] (n2_1) at (2,1) {};
\draw[-latex,darkred] (n1_1) to[out=45,in=135] (n2_1);
\draw[-latex,darkblue] (n2_1) to[out=225,in=315] (n1_1);
\draw[-latex,darkred] (n2_0) -- (n2_1);
\draw[-latex,very thick,darkgreen] (n1_0) -- (n2_1);
\node[inner sep = 1pt,circle,fill=black] (n2_2) at (2,2) {};
\draw[-latex,darkred] (n1_2) to[out=45,in=135] (n2_2);
\draw[-latex,darkblue] (n2_2) to[out=225,in=315] (n1_2);
\draw[-latex,darkred] (n2_1) -- (n2_2);
\node[inner sep = 1pt,circle,fill=black] (n2_3) at (2,3) {};
\draw[-latex,darkred] (n1_3) to[out=45,in=135] (n2_3);
\draw[-latex,darkblue] (n2_3) to[out=225,in=315] (n1_3);
\draw[-latex,darkred] (n2_2) -- (n2_3);
\node[inner sep = 1pt,circle,fill=black] (n2_4) at (2,4) {};
\draw[-latex,darkred] (n1_4) to[out=45,in=135] (n2_4);
\draw[-latex,darkblue] (n2_4) to[out=225,in=315] (n1_4);
\draw[-latex,darkred] (n2_3) -- (n2_4);
\node[inner sep = 1pt,circle,fill=black] (n2_5) at (2,5) {};
\draw[-latex,darkred] (n1_5) to[out=45,in=135] (n2_5);
\draw[-latex,darkblue] (n2_5) to[out=225,in=315] (n1_5);
\draw[-latex,darkred] (n2_4) -- (n2_5);
\node[inner sep = 1pt,circle,fill=black] (n3_0) at (3,0) {};
\draw[-latex,darkred] (n2_0) to[out=45,in=135] (n3_0);
\draw[-latex,darkblue] (n3_0) to[out=225,in=315] (n2_0);
\node[inner sep = 1pt,circle,fill=black] (n3_1) at (3,1) {};
\draw[-latex,darkred] (n2_1) to[out=45,in=135] (n3_1);
\draw[-latex,darkblue] (n3_1) to[out=225,in=315] (n2_1);
\draw[-latex,darkred] (n3_0) -- (n3_1);
\node[inner sep = 1pt,circle,fill=black] (n3_2) at (3,2) {};
\draw[-latex,darkred] (n2_2) to[out=45,in=135] (n3_2);
\draw[-latex,darkblue] (n3_2) to[out=225,in=315] (n2_2);
\draw[-latex,darkred] (n3_1) -- (n3_2);
\node[inner sep = 1pt,circle,fill=black] (n3_3) at (3,3) {};
\draw[-latex,darkred] (n2_3) to[out=45,in=135] (n3_3);
\draw[-latex,darkblue] (n3_3) to[out=225,in=315] (n2_3);
\draw[-latex,darkred] (n3_2) -- (n3_3);
\node[inner sep = 1pt,circle,fill=black] (n3_4) at (3,4) {};
\draw[-latex,darkred] (n2_4) to[out=45,in=135] (n3_4);
\draw[-latex,darkblue] (n3_4) to[out=225,in=315] (n2_4);
\draw[-latex,darkred] (n3_3) -- (n3_4);
\draw[-latex,very thick,darkgreen] (n2_3) -- (n3_4);
\node[inner sep = 1pt,circle,fill=black] (n3_5) at (3,5) {};
\draw[-latex,darkred] (n2_5) to[out=45,in=135] (n3_5);
\draw[-latex,darkblue] (n3_5) to[out=225,in=315] (n2_5);
\draw[-latex,darkred] (n3_4) -- (n3_5);
\node[inner sep = 1pt,circle,fill=black] (n4_0) at (4,0) {};
\draw[-latex,darkred] (n3_0) to[out=45,in=135] (n4_0);
\draw[-latex,darkblue] (n4_0) to[out=225,in=315] (n3_0);
\node[inner sep = 1pt,circle,fill=black] (n4_1) at (4,1) {};
\draw[-latex,darkred] (n3_1) to[out=45,in=135] (n4_1);
\draw[-latex,darkblue] (n4_1) to[out=225,in=315] (n3_1);
\draw[-latex,darkred] (n4_0) -- (n4_1);
\node[inner sep = 1pt,circle,fill=black] (n4_2) at (4,2) {};
\draw[-latex,darkred] (n3_2) to[out=45,in=135] (n4_2);
\draw[-latex,darkblue] (n4_2) to[out=225,in=315] (n3_2);
\draw[-latex,darkred] (n4_1) -- (n4_2);
\node[inner sep = 1pt,circle,fill=black] (n4_3) at (4,3) {};
\draw[-latex,darkred] (n3_3) to[out=45,in=135] (n4_3);
\draw[-latex,darkblue] (n4_3) to[out=225,in=315] (n3_3);
\draw[-latex,darkred] (n4_2) -- (n4_3);
\node[inner sep = 1pt,circle,fill=black] (n4_4) at (4,4) {};
\draw[-latex,darkred] (n3_4) to[out=45,in=135] (n4_4);
\draw[-latex,darkblue] (n4_4) to[out=225,in=315] (n3_4);
\draw[-latex,darkred] (n4_3) -- (n4_4);
\node[inner sep = 1pt,circle,fill=black] (n4_5) at (4,5) {};
\draw[-latex,darkred] (n3_5) to[out=45,in=135] (n4_5);
\draw[-latex,darkblue] (n4_5) to[out=225,in=315] (n3_5);
\draw[-latex,darkred] (n4_4) -- (n4_5);
\node[inner sep = 1pt,circle,fill=black] (n5_0) at (5,0) {};
\draw[-latex,darkred] (n4_0) to[out=45,in=135] (n5_0);
\draw[-latex,darkblue] (n5_0) to[out=225,in=315] (n4_0);
\node[inner sep = 1pt,circle,fill=black] (n5_1) at (5,1) {};
\draw[-latex,darkred] (n4_1) to[out=45,in=135] (n5_1);
\draw[-latex,darkblue] (n5_1) to[out=225,in=315] (n4_1);
\draw[-latex,darkred] (n5_0) -- (n5_1);
\node[inner sep = 1pt,circle,fill=black] (n5_2) at (5,2) {};
\draw[-latex,darkred] (n4_2) to[out=45,in=135] (n5_2);
\draw[-latex,darkblue] (n5_2) to[out=225,in=315] (n4_2);
\draw[-latex,darkred] (n5_1) -- (n5_2);
\node[inner sep = 1pt,circle,fill=black] (n5_3) at (5,3) {};
\draw[-latex,darkred] (n4_3) to[out=45,in=135] (n5_3);
\draw[-latex,darkblue] (n5_3) to[out=225,in=315] (n4_3);
\draw[-latex,darkred] (n5_2) -- (n5_3);
\draw[-latex,very thick,darkgreen] (n4_2) -- (n5_3);
\node[inner sep = 1pt,circle,fill=black] (n5_4) at (5,4) {};
\draw[-latex,darkred] (n4_4) to[out=45,in=135] (n5_4);
\draw[-latex,darkblue] (n5_4) to[out=225,in=315] (n4_4);
\draw[-latex,darkred] (n5_3) -- (n5_4);
\node[inner sep = 1pt,circle,fill=black] (n5_5) at (5,5) {};
\draw[-latex,darkred] (n4_5) to[out=45,in=135] (n5_5);
\draw[-latex,darkblue] (n5_5) to[out=225,in=315] (n4_5);
\draw[-latex,darkred] (n5_4) -- (n5_5);
\node[inner sep = 1pt,circle,fill=black] (n6_0) at (6,0) {};
\draw[-latex,darkred] (n5_0) to[out=45,in=135] (n6_0);
\draw[-latex,darkblue] (n6_0) to[out=225,in=315] (n5_0);
\node[inner sep = 1pt,circle,fill=black] (n6_1) at (6,1) {};
\draw[-latex,darkred] (n5_1) to[out=45,in=135] (n6_1);
\draw[-latex,darkblue] (n6_1) to[out=225,in=315] (n5_1);
\draw[-latex,darkred] (n6_0) -- (n6_1);
\node[inner sep = 1pt,circle,fill=black] (n6_2) at (6,2) {};
\draw[-latex,darkred] (n5_2) to[out=45,in=135] (n6_2);
\draw[-latex,darkblue] (n6_2) to[out=225,in=315] (n5_2);
\draw[-latex,darkred] (n6_1) -- (n6_2);
\node[inner sep = 1pt,circle,fill=black] (n6_3) at (6,3) {};
\draw[-latex,darkred] (n5_3) to[out=45,in=135] (n6_3);
\draw[-latex,darkblue] (n6_3) to[out=225,in=315] (n5_3);
\draw[-latex,darkred] (n6_2) -- (n6_3);
\node[inner sep = 1pt,circle,fill=black] (n6_4) at (6,4) {};
\draw[-latex,darkred] (n5_4) to[out=45,in=135] (n6_4);
\draw[-latex,darkblue] (n6_4) to[out=225,in=315] (n5_4);
\draw[-latex,darkred] (n6_3) -- (n6_4);
\node[inner sep = 1pt,circle,fill=black] (n6_5) at (6,5) {};
\draw[-latex,darkred] (n5_5) to[out=45,in=135] (n6_5);
\draw[-latex,darkblue] (n6_5) to[out=225,in=315] (n5_5);
\draw[-latex,darkred] (n6_4) -- (n6_5);
\draw[-latex,very thick,darkgreen] (n5_4) -- (n6_5);
\draw (0, -0.25) node[below]{$\mathtt{0}$};
\draw (1, -0.25) node[below]{$\mathtt{1}$};
\draw (2, -0.25) node[below]{$\mathtt{2}$};
\draw (3, -0.25) node[below]{$\mathtt{3}$};
\draw (4, -0.25) node[below]{$\mathtt{4}$};
\draw (5, -0.25) node[below]{$\mathtt{5}$};
\draw (6, -0.25) node[below]{$\mathtt{6}$};
\draw (-0.25,0) node[left]{$\mathtt{0}$};
\draw (-0.25,1) node[left]{$\mathtt{1}$};
\draw (-0.25,2) node[left]{$\mathtt{2}$};
\draw (-0.25,3) node[left]{$\mathtt{3}$};
\draw (-0.25,4) node[left]{$\mathtt{4}$};
\draw (-0.25,5) node[left]{$\mathtt{5}$};
            \end{tikzpicture}
        \end{center}
    
        \caption{The alignment graph for $s = [1, 0, 3, \plh, 2, 4]$, with red edges of weight $0$, green edges of weight $1$, and blue edges of weight $-2$.}
        \label{fig:alignment-graph}
    \end{figure}
\end{definition}

We now show that the longest paths in $G^s$ represent the answers to the range LIS queries on $s$.

\begin{lemma} \label{lm:no-positive-cycles}
    For a subpermutation $s$, the graph $G^s$ does not contain any cycles of nonnegative length.
\end{lemma}

\begin{proof}
    Consider some cycle $C$ in $G^s$.
    As $G^s$ does not contain any edges that decrease the second coordinate of a vertex, no cycle in $G^s$ may contain edges of the first and the third types from \cref{def:alignment-graph} as they increase the second coordinate of a vertex.
    As no cycle can consist of only edges of the second type, every cycle contains at least one edge of the fourth type and thus has a strictly negative length.
\end{proof}

\begin{definition}
    For a subpermutation $s$ of length $n$, define the \emph{distance matrix} $M^s$ of size $(n + 1) \times (n + 1)$ such that $M^s_{i, j}$ for any $i, j \in \fragmentcc{0}{n}$ is equal to the length of the longest path in $G^s$ from $(i, 0)$ to $(j, |s^*|)$.
%\lipicsEnd
\end{definition}
    
Note that the matrix $M^s$ is well-defined because, for every $i, j \in \fragmentcc{0}{n}$, there is a path from $(i, 0)$ to $(j, |s^*|)$ in $G^s$, and since $G^s$ does not contain any cycles of nonnegative length by \cref{lm:no-positive-cycles}, there exists a longest such path.
Furthermore, the matrix $M^s$ is anti-Monge by the following fact. (If we negate the weights in $G^s$, we obtain that the matrix $-M^s$ of the shortest paths in such a graph is Monge.)

\begin{fact}[{\cite[Section 2.3]{FR06}}]\label{fct:monge-planar}
    Consider a weighted directed planar graph $G$. 
    For vertices $u_0,\ldots,u_{p-1},v_{q-1},\allowbreak \ldots,v_0$ lying (in this cyclic order, potentially with repetitions) on the outer face of $G$, define a $p\times q$ matrix $D$ with $D_{i,j}=\dist_G(u_i,v_j)$, where $\dist_G$ denotes the shortest path distance in $G$.
    If all entries of $D$ are finite, then $D$ is a Monge matrix.
%%\lipicsEnd
\end{fact}

\begin{figure}[hbt]
    \begin{center}
        \includegraphics[scale=0.5]{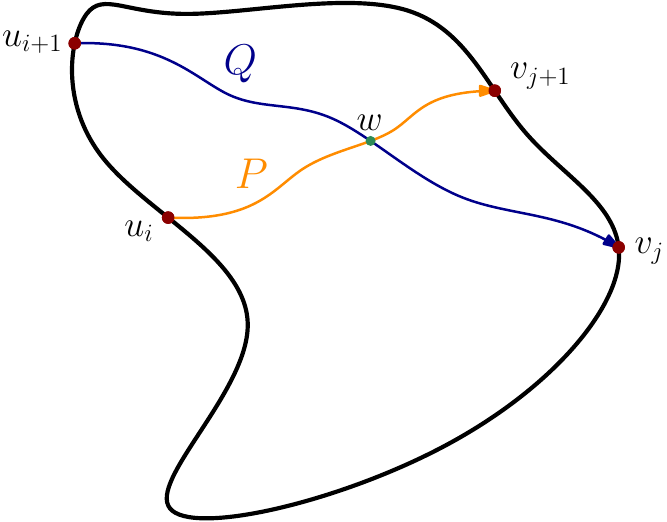}
    \end{center}

    \caption{The sum of the distances from $u_i$ to $v_{j + 1}$ and from $u_{i + 1}$ to $v_{j}$ is not smaller than the sum of the distances from $u_i$ to $v_j$ and from $u_{i + 1}$ to $v_{j + 1}$.}
    \label{fig:monge-planar}
\end{figure}

\begin{proof}[Proof Sketch]
    We need to prove $D_{i, j} + D_{i + 1, j + 1} \le D_{i, j + 1} + D_{i + 1, j}$ for $i \in \fragmentco{0}{p-1}$ and $j \in \fragmentco{0}{q-1}$.
    Note that $D_{i, j + 1}=\dist_G(u_i,v_{j+1})$ and $D_{i + 1, j}=\dist_G(u_{i+1},v_j)$.
    Fix a shortest path $P$ from $u_i$ to $v_{j + 1}$ and a shortest path $Q$ from $u_{i + 1}$ to $v_j$; see \cref{fig:monge-planar}.
    As these four vertices lie in the order $u_i, u_{i + 1}, v_{j + 1}, v_j$ on the outer face of $G$ (some of these vertices may coincide with each other), the paths $P$ and $Q$ intersect at a common vertex $w \in V(G)$.
    As $P$ and $Q$ are shortest paths, we have
    \begin{align*}
        D_{i, j + 1}\ + D_{i + 1, j} &= (\dist_G(u_i, w) + \dist_G(w, v_{j + 1})) + (\dist_G(u_{i + 1}, w) + \dist_G(w, v_j))\\
                                     &= (\dist_G(u_i, w) + \dist_G(w, v_{j})) + (\dist_G(u_{i + 1}, w) + \dist_G(w, v_{j + 1}))\\
                                     &\ge D_{i, j} + D_{i + 1, j + 1}.\qedhere
    \end{align*}
\end{proof}

We now show that the matrix $M^s$ contains all range LIS values.

\begin{lemma}\label{lm:ms-meaning}
    For a subpermutation $s$ of length $n$ and any $i, j \in \fragmentcc{0}{n}$, we have
    \[
        M^s_{i, j} = 
        \begin{cases}
            \lis(s\fragmentco{i}{j}), & \text{if $i < j$;}\\
            -2 \cdot (i - j), & \text{if $i \ge j$.}
        \end{cases}
    \]
    Furthermore, all corresponding longest paths are monotone in both coordinates.
\end{lemma}

\begin{proof}
    First, assume that $i \ge j$.
    Consider some path $P$ from $(i, 0)$ to $(j, |s^*|)$.
    Let $k$ be the number of edges of the third type from \cref{def:alignment-graph} on this path, and $\ell$ be the number of edges of the fourth type on this path.
    As edges of the first two types have weights zero, the length of $P$ is $k - 2 \ell$.
    Note that the first coordinate of the vertex along the path $P$ is increased by at least $k$ and decreased by at least $(i - j) + k$.
    Only edges of the fourth type decrease the first coordinate of the current vertex, so $\ell \ge (i - j) + k$.
    Therefore, the length of $P$ is at most $k - 2 \cdot ((i - j) + k) = -2 \cdot (i - j) - k \le -2 \cdot (i - j)$.
    Furthermore, there exists a path of length exactly $-2 \cdot (i - j)$: it first takes vertical edges of weight $0$ from $(i, 0)$ to $(i, |s^*|)$, and then horizontal backward edges of weight $-2$ from $(i, |s^*|)$ to $(j, |s^*|)$.
    Therefore, $M^S_{i, j} = -2 \cdot (i - j)$.
    The only way to achieve this length is for $k = 0$, and thus any such longest path is monotone in both coordinates.

    We now consider the case $i < j$.
    Let $P$ be the longest path from $(i, 0)$ to $(j, |s^*|)$.
    We first show that $P$ is monotone in both coordinates.
    As there are no edges in $G^s$ that decrease the second coordinate of a vertex, $P$ is monotone in the second coordinate.
    It remains to show that $P$ is monotone in the first coordinate.
    For the sake of contradiction, suppose that it is not, that is, there exist two vertices $(x, y_1), (x, y_2) \in P$ such that between them the vertices of $P$ have the first coordinates not equal to $x$.
    Similarly to the case $i \ge j$ (in particular, $i = j$), any such subpath of $P$ from $(x, y_1)$ to $(x, y_2)$ has non-positive length.
    Furthermore, the only path of length zero is the vertical path $(x, y_1) \to (x, y_1 + 1) \to \cdots \to (x, y_2)$.
    Hence, $P$ is not the longest path, thus leading to a contradiction.
    Therefore, indeed, every longest path $P$ from $(i, 0)$ to $(j, |s^*|)$ is monotone in both dimensions.
    Consequently, $P$ only uses the edges of the first three types from \cref{def:alignment-graph}, and thus its length is equal to the number of edges of the third type that it uses.
    It is easy to see that there is a bijection between the increasing subsequences of $s\fragmentco{i}{j}$ and collections of edges of the third type that may lie on a path from $(i, 0)$ to $(j, |s^*|)$.
    Therefore, $M^s_{i, j} = \lis(s\fragmentco{i}{j})$.
\end{proof}

\begin{corollary} \label{lm:ms-no-core-below-diagonal}
    For a subpermutation $s$ and a core element $(i, j, v) \in \core(M^s)$, we have $i \le j$.
\end{corollary}

\begin{proof}
    For $i > j$, all elements in $M^s\fragmentcc{i}{i + 1}\fragmentcc{j}{j+1}$ satisfy the conditions for the second case of \cref{lm:ms-meaning}, and thus $M^{s \square}_{i, j} = 0$.
\end{proof}

\begin{observation}\label{lm:simple-core-bound}
    For an integer-valued Monge matrix $A$, we have $\delta(A) \le \deltasigma(A)$.
\end{observation}
\begin{proof}
    Each value in $\core(A)$ is a positive integer and contributes at least one to $\deltasigma(A)$.
\end{proof}

\begin{corollary}\label{lm:ms-core-bound}
    For a subpermutation $s$ of length $n$, we have $\delta(M^s) = \Oh(n)$.
\end{corollary}

\begin{proof}
    Due to \cref{lm:ms-meaning}, all values of $M^s$ are in $\fragmentcc{-2n}{n}$.
    Due to \cref{lm:calculate-value-from-core}, we have $-\deltasigma(M^s) = M^s_{0, 0} + M^s_{n, n} - M^s_{0, n}- M^s_{n, 0} = 2n-\lis(s)= \Oh(n)$.
    Therefore, we have $\delta(M^s) = \Oh(n)$ due to \cref{lm:simple-core-bound}.\footnote{As $M^s$ is anti-Monge, \cref{lm:simple-core-bound} implies $\delta(M^s) \le -\deltasigma(M^s)$.}
\end{proof}

As implied by \cref{lm:ms-meaning}, to answer range LIS value queries on some permutation $s$, it is sufficient to obtain random access to $M^s$.
We first describe a divide-and-conquer procedure that obtains the condensed representation of $M^s$ and then apply \cref{lm:cmo} to convert it into an efficient oracle.
Before we describe the complete divide-and-conquer procedure, we prove a helper lemma that will be required for the conquer step.

\begin{lemma}\label{lm:lis-decompresser}
    There is an algorithm that, given a subpermutation $s$, a subsequence $t$ of $s$ with $s^* = t^*$, and the condensed representation of $M^{t}$, in time $\Oh(|s|)$ computes the condensed representation of $M^s$.
\end{lemma}

\begin{proof}
    Let $n \coloneqq |s|$ and $m \coloneqq |t|$.
    As $t$ is a subsequence of $s$, a linear-time greedy algorithm computes an increasing injection of $t$ onto $s$.
    That is, for each $j \in \fragmentco{0}{m}$, we compute $\pos(j)$ such that $t\position{j} = s\position{\pos(j)}$ and $\pos(j - 1) < \pos(j)$ for $j \in \fragmentoo{0}{m}$.
    Furthermore, we set $\pos(m) = n$ for convenience.
    Moreover, for each $i \in \fragmentcc{0}{n}$, we compute $\rpos(i)$ as the smallest index $j \in \fragmentcc{0}{m}$ with $\pos(j) \ge i$.
    These mappings allow us to convert the values of $M^s$ to the values of $M^{t}$:
    For $i, j \in \fragmentcc{0}{n}$ with $i \le j$, we have $M^s_{i, j} = M^{t}_{\rpos(i), \rpos(j)}$ because (due to $s^* = t^*$) $s$ is only different from $t$ in $\plh$ elements.

    We now describe how to obtain the values in the topmost row and the leftmost column of $M^s$.
    The values in the leftmost column can be obtained trivially in linear time due to \cref{lm:ms-meaning}.
    That is, $M^s_{i, 0} = -2i$ for $i \in \fragmentcc{0}{n}$.
    Furthermore, we have $M^{s}_{0, j} = M^{t}_{\rpos(0), \rpos(j)} = M^{t}_{0, \rpos(j)}$ for all $j \in \fragmentcc{0}{n}$.
    Hence, having the values of the topmost row of $M^t$, we can compute the values of the topmost row of $M^s$ in linear time.
    Therefore, it remains to compute the core of $M^s$.

    \smallskip

    For conciseness, denote $A \coloneqq M^s$ and $B \coloneqq M^t$.
    Note that $s$ can be obtained from $t$ by repeatedly inserting $\plh$ at some positions.
    We first consider the case $m = n - 1$, when $s$ is obtained from $t$ by inserting $\plh$ at a single position $k \in \fragmentcc{0}{m}$.
    Then, due to \cref{lm:ms-meaning}, the following equalities are trivial:
    $A\fragmentcc{0}{k}\fragmentcc{0}{k} = B\fragmentcc{0}{k}\fragmentcc{0}{k}$,
    $A\fragmentcc{0}{k}\fragmentoc{k}{n} = B\fragmentcc{0}{k}\fragmentcc{k}{m}$,
    and $A\fragmentoc{k}{n}\fragmentoc{k}{n} = B\fragmentcc{k}{m}\fragmentcc{k}{m}$.

    Thus, it follows that $\dens{A}\fragmentco{0}{k}\fragmentco{0}{k} = \dens{B}\fragmentco{0}{k}\fragmentco{0}{k}$, $\dens{A}\fragmentco{0}{k}\fragmentoo{k}{n} = \dens{B}\fragmentco{0}{k}\fragmentco{k}{m}$, and $\dens{A}\fragmentoo{k}{n}\fragmentoo{k}{n} = \dens{B}\fragmentco{k}{m}\fragmentco{k}{m}$.
    Furthermore, as $A\fragmentcc{0}{k}\fragmentcc{k+1}{k+1} = A\fragmentcc{0}{k}\fragmentcc{k}{k}$ due to \cref{lm:ms-meaning}, we have that $\dens{A}\fragmentco{0}{k}\fragmentcc{k}{k}$ consists of zeros.
    Analogously, as $A\fragmentcc{k+1}{k+1}\fragmentoc{k}{n} = A\fragmentcc{k}{k}\fragmentoc{k}{n}$, we have that $\dens{A}\fragmentcc{k}{k}\fragmentoo{k}{n}$ consists of zeros.
    Moreover, due to \cref{lm:ms-no-core-below-diagonal}, $\dens{A}_{i, j} = 0$ for all $i, j \in \fragmentco{0}{n}$ with $i > j$.
    The only element of $\dens{A}$ that was not yet considered is $\dens{A}_{k, k} = A_{k, k + 1} + A_{k + 1, k} - A_{k, k} - A_{k + 1, k + 1} = \lis(\{\plh\}) + (-2) - 0 - 0 = -2$ due to \cref{lm:ms-meaning}.

    Putting all these facts together, we derive that every core element $(i, j, v) \in \core(B)$ is replaced in $\core(A)$ by $(i, j, v)$ if $i, j \in \fragmentco{0}{k}$, by $(i, j + 1, v)$ if $i \in \fragmentco{0}{k}$ and $j \in \fragmentco{k}{m}$, and by $(i + 1, j + 1, v)$ if $i, j \in \fragmentco{k}{m}$.
    Furthermore, one new core element $(k, k, -2)$ is added to $\core(A)$.

    In other words, for a core element of $B$, its first coordinate is incremented if it is not smaller than $k$ and its second coordinate is incremented if it is not smaller than $k$.
    Therefore, alternatively we may say that $(i, j, v) \in \core(B)$ is replaced by $(\pos(i), \pos(j), v)$ in $\core(A)$.
    Note that if we now drop the condition $m = n - 1$ and consider a collection of insertions of $\plh$ into $t$ that transforms it into $s$, the same property holds.
    Consequently, we have $\core(A) = \{(\pos(i), \pos(j), v) \mid (i, j, v) \in \core(B) \} \cup \{(k, k, -2) \mid k \in \fragmentco{0}{n} \setminus \{\pos(i) \mid i \in \fragmentco{0}{m}\}\}$.
    Thus, $\core(A)$ can be computed in linear time as $\delta(B) = \Oh(m) \le \Oh(n)$ due to \cref{lm:ms-core-bound}.
\end{proof}

\begin{lemma}\label{lm:lis-dc}
    There is an algorithm that, given a permutation $s$ of length $n$, in time $\Oh(n \log^2 n)$ computes the condensed representation of $M^s$.
\end{lemma}

\begin{proof}
    Let the initial permutation be $s_0$, and its length be $n_0$.
    We create a divide-and-conquer recursive procedure that computes $M^{s_0}$.
    Let the permutation in the current recursive call be $s$, and its length be $n$.
    If $n = 1$, we construct $M^s$ trivially.
    Otherwise, let $m \coloneqq \floor{\frac{n}{2}}$.
    We construct two subpermutations $s^{\lo}$ and $s^{\hi}$.
    For all $i \in \fragmentco{0}{n}$, we define $s^{\lo}\position{i} = s\position{i}$ and $s^{\hi}\position{i} = \plh$ if $s_i < m$, and $s^{\lo}\position{i} = \plh$ and $s^{\hi}\position{i} = s\position{i} - m$ otherwise.
    We compute the condensed representations of $M^{s^{\lo *}}$ and $M^{s^{\hi *}}$ recursively and use \cref{lm:lis-decompresser} to obtain the condensed representations of $M^{s^{\lo}}$ and $M^{s^{\hi}}$.
    We claim that $M^s$ is the $(\max, +)$-product of $M^{s^{\lo}}$ and $M^{s^{\hi}}$ and compute its condensed representation using \cref{lm:core-based-multiplication-algorithm}.\footnote{The $(\max, +)$-product can be computed as $-((-M^{s^{\lo}}) \minplus (-M^{s^{\hi}}))$, and to get the condensed representation of a negated matrix, we just need to negate all values in the topmost row, the leftmost column, and the core.}

    Let us prove correctness.
    Define the sets $S_j \coloneqq \fragmentcc{0}{n} \times \fragmentcc{j}{j}$ of vertices of $G^s$ for each $j \in \fragmentcc{0}{n}$.
    Note that $M^s$ is the matrix of the lengths of the longest paths from $S_0$ to $S_n$.
    Furthermore, as no edges in $G^s$ decrease the second coordinate of a vertex, every path from $S_0$ to $S_n$ crosses $S_m$ exactly once.
    Moreover, note that $M^{s^{\lo}}$ is the matrix of the lengths of the longest paths from $S_0$ to $S_m$ and $M^{s^{\hi}}$ is the matrix of the lengths of the longest paths from $S_m$ to $S_n$ in $G^s$.
    Therefore, $M^s$ is the $(\max, +)$-product of $M^{s^{\lo}}$ and $M^{s^{\hi}}$, and the algorithm is correct.

    We now analyze the running time of the algorithm.
    If $n = 1$, the running time is constant.
    Otherwise, the running time of a single recursive call is dominated by the call to \cref{lm:core-based-multiplication-algorithm}.
    Due to \cref{lm:ms-core-bound}, we have $\delta(M^{s^{\lo}}), \delta(M^{s^{\hi}}) = \Oh(n)$, and thus the application of \cref{lm:core-based-multiplication-algorithm} takes $\Oh(n \log n)$ time.
    Each recursive step splits $n$ approximately in half, so the depth of the recursion is $\Oh(\log n_0)$.
    At each level, the total time spent is $\Oh(n_0 \log n_0)$.
    Therefore, the total time complexity of the algorithm is $\Oh(n_0 \log^2 n_0)$.
\end{proof}

We are now ready to replicate the result of \cite{Tiskin10}.

\begin{lemma}\label{lm:lis-no-reporting}
    There is an algorithm that, given a permutation $s$ of length $n$, in time $\Oh(n \log^2 n)$ builds a data structure that can answer range LIS value queries in $\Oh(\log n)$ time.
\end{lemma}

\begin{proof}
    We apply \cref{lm:lis-dc} to obtain the condensed representation of $M^s$ in time $\Oh(n \log^2 n)$.
    We then spend additional $\Oh(n \log n)$ time to compute $\mds(M^s)$ using \cref{lm:cmo}, where $\delta(M^s) = \Oh(n)$ due to \cref{lm:ms-core-bound}.
    Now whenever we get a query for $\lis(s\fragmentco{i}{j})$, we use $\mds(M^s)$ to obtain $M^s_{i, j}$ in time $\Oh(\log n)$, which is indeed the answer due to \cref{lm:ms-meaning}.
\end{proof}

We now extend the algorithm of \cref{lm:lis-no-reporting} to answer range LIS reporting queries.
This improves upon the result of \cite{KR24}.

\thmlis*

\begin{proof}
    We first run the algorithm of \cref{lm:lis-dc} and store the condensed representations of all intermediate distance matrices and $\pos$ and $\rpos$ correspondences between sequences $s$ and $s^*$.
    Furthermore, we replace all calls to \cref{lm:core-based-multiplication-algorithm} in \cref{lm:lis-dc} with the calls to \cref{lm:product-witness-reconstruction}.
    There is a total of $2n - 1$ recursive calls in this algorithm.\footnote{A full binary tree with $n$ leaves has $2n-1$ nodes.}
    We arbitrarily index them with numbers from $0$ to $2n-2$.
    Denote the root recursive call by $\rho$, and for every internal recursive call $\nu$, denote its children recursive calls by $\nu_L$ and $\nu_R$.
    Let $s_\nu$ be the input permutation of the $\nu$-th recursive call for each $\nu \in \fragmentco{0}{2n-1}$.
    Denote $\tau \coloneqq \ceil{\log^{2 - \alpha} n} + 1$.

    We compute the dynamic programming table with values $dp(\nu, i, \ell)$ for all $\nu \in \fragmentco{0}{2n-1}$, $i \in \fragmentco{0}{|s_\nu|}$, and $\ell \in \fragmentcc{1}{\tau}$.
    Here, $dp(\nu, i, \ell)$ is the smallest index $j \in \fragmentoc{i}{|s_\nu|}$ such that $\lis(s_\nu\fragmentco{i}{j}) \ge \ell$ or $|s_\nu| + 1$ if no such $j$ exists.
    Furthermore, if $dp(\nu, i, \ell) \neq |s_\nu| + 1$, we additionally store the information sufficient to reconstruct the corresponding increasing subsequence of $s_{\nu}\fragmentco{i}{dp(\nu, i, \ell)}$ of length $\ell$ in $\Oh(\ell)$ time (call this additional information $dp'$).
    If $\ell = 1$, we explicitly store the only element of the subsequence of length $1$.
    Otherwise, if $\ell > 1$, we store two dynamic programming states $(\nu_1, i_1, \ell_1)$ and $(\nu_2, i_2, \ell_2)$ with $\ell_1, \ell_2 \ge 1$ and $\ell_1 + \ell_2 = \ell$ such that the corresponding subsequence for $(\nu, i, \ell)$ is the concatenation of the corresponding subsequences for $(\nu_1, i_1, \ell_1)$ and $(\nu_2, i_2, \ell_2)$.
    This way, we can indeed recursively reconstruct the subsequence in time $\Oh(\ell)$.
    We now show how to compute the values of $dp$ and $dp'$.

    \begin{claim}
        The values of $dp$ and $dp'$ can be computed in time $\Oh(n \tau \log n)$.
    \end{claim}

    \begin{claimproof}
        We compute the dynamic programming tables in the bottom-up order of the recursion tree.
        That is, we assume that, when we compute $dp(\nu, \cdot, \cdot)$, the values $dp(\nu_L, \cdot, \cdot)$ and $dp(\nu_R, \cdot, \cdot)$ are already available.
        Fix some recursion node  $\nu$ and let $n_{\nu} \coloneqq |s_{\nu}|$.
        We iterate over $\ell \in \fragmentcc{1}{\tau}$ and compute all values $dp(\nu, \cdot, \ell)$ and $dp'(\nu, \cdot, \ell)$.
        For $\ell = 1$, we compute the values trivially:
        Any nonempty sequence has an increasing subsequence of length one, so $dp(\nu, i, 1) = i + 1$ and $dp'(\nu, i, 1) = s_{\nu}\position{i}$ for all $i \in \fragmentco{0}{|s_{\nu}|}$.\footnote{Actually, $s_{\nu}\position{i}$ is not equal to the value from the initial permutation $s$. 
        To circumvent this issue, for all elements of the permutations in the recursion, we remember the values from the initial permutation they correspond to.}
        
        We now consider the case $\ell > 1$.
        For conciseness, let $A \coloneqq M^{s^{\lo}_{\nu}}, B \coloneqq M^{s^{\hi}_{\nu}}$, and $C \coloneqq M^{s_{\nu}}$.
        Note that $dp(\nu, i, \ell)$ is non-decreasing over $i$.
        We find all these values using three nondecreasing pointers.
        We maintain three values $i, j$, and $k$, where $k = dp(\nu, i, \ell)$, and $j$ is the smallest value such that $(j, m)$ is on some longest path from $(i, 0)$ to $(k, n_{\nu})$ in $G^{s_{\nu}}$, where $m = \floor{\frac{n_{\nu}}{2}}$ is taken from \cref{lm:lis-dc}.
        In other words, $j = \wit{A}{B}_{i, k}$.
        As throughout the algorithm's lifetime, $i$ and $k$ are only increasing, $j$ is also nondecreasing due to \cref{lm:witness-monotonicity}.

        We have that $C$ is the $(\max, +)$-product of $A$ and $B$.
        Initially we build $\lco(A), \lco(B)$, and $\lco(C)$ using \cref{lm:local-core-oracle}.
        Furthermore, we initially set $i = j = k = 0$.
        Throughout the lifetime of the algorithm, we maintain the values $A_{i, j}, B_{j, k}$, and $C_{i, k}$.
        Initially, we compute them using boundary access of $\lco$.
        At each step, for fixed $i$ and $k$, we increment $j$ as long as $A_{i, j} + B_{j, k} \neq C_{i, k}$, that is, until $j$ becomes the smallest witness.
        Furthermore, for a fixed $i$, we increment $k$ as long as $C_{i, k} < \ell$, that is, until $k$ becomes equal to $dp(\nu, i, \ell)$.
        If we still have $C_{i, k} < \ell$ for $k = n_{\nu}$, then we stop the algorithm and set $dp(\nu, i', \ell) = n_{\nu} + 1$ for all $i' \in \fragmentco{i}{n_{\nu}}$.

        When we find the right values $j$ and $k$ for the given $i$, we set $dp(\nu, i, \ell) \coloneqq k$ and compute $dp'(\nu, i, \ell)$ as described below.
        In constant time we can compute the indices $i^{\lo}$ and $j^{\lo}$ in $s^{\lo *}_{\nu}$ corresponding to the indices $i$ and $j$ in $s^{\lo}_{\nu}$ using the $\rpos$ mapping for $s^{\lo}_{\nu}$.
        Similarly, in constant time we can compute the indices $j^{\hi}$ and $k^{\hi}$ in $s^{\hi *}_{\nu}$ corresponding to the indices $j$ and $k$ in $s^{\hi}_{\nu}$.
        If $A_{i, j} = 0$, all the elements of the longest increasing subsequence lie in $s^{\hi}_{\nu}$, and thus we set $dp'(\nu, i, \ell) \coloneqq dp'(\nu_R, j^{\hi}, \ell)$.
        Analogously, if $B_{j, k} = 0$, all the elements of the longest increasing subsequence lie in $s^{\lo}_{\nu}$, and we set $dp'(\nu, i, \ell) \coloneqq dp'(\nu_L, i^{\hi}, \ell)$.
        Finally, if $A_{i, j} > 0$ and $B_{j, k} > 0$, the longest increasing subsequence is split across $s^{\lo}_{\nu}$ and $s^{\hi}_{\nu}$, and thus we set $dp'(\nu, i, \ell) \coloneqq ((\nu_L, i^{\lo}, A_{i, j}), (\nu_R, j^{\hi}, B_{j, k}))$.

        After computing $dp(\nu, i, \ell)$ and $dp'(\nu, i, \ell)$, we increment $i$ and continue the algorithm.
        At each step of the algorithm, we increment one of the variables $i, j$, or $k$ and recompute the corresponding values $A_{i, j}$, $B_{j, k}$, and $C_{i, k}$ using the horizontally or vertically adjacent recomputation of $\lco$.

        We now analyze the time complexity of the algorithm.
        First, we analyze it for a fixed recursion node $\nu$.
        For $\ell = 1$, it takes $\Oh(n_{\nu})$ time to compute $dp(\nu, \cdot, \ell)$ and $dp'(\nu, \cdot, \ell)$.
        For any other $\ell$, we go through all indices $i \in \fragmentco{0}{n_{\nu}}$, and, for each one of them, find the corresponding $j$ and $k$ and then compute the dynamic programming values.
        All operations except for finding $j$ and $k$ take constant time.
        The index $k$ is increased by at most $n_{\nu}$ across all iterations.
        Whenever it is increased, we spend $\Oh(\deltaj{k}(C) + 1)$ time to recompute $C_{i, k}$, which sums up to $\Oh(\delta(C) + n_{\nu})$ time in total.
        Furthermore, we spend $\Oh(\deltaj{k}(B) + 1)$ time to recompute $B_{j, k}$, which sums up to $\Oh(\delta(B) + n_{\nu})$ time in total.
        Analogously, the index $j$ is increased by at most $n_{\nu}$ across all iterations, and we spend at most $\Oh(n_{\nu} + \delta(A) + \delta(B))$ time in total on it.
        Finally, the index $i$ is increased by at most $n_{\nu}$ across all iterations, and we spend at most $\Oh(n_{\nu} + \delta(A) + \delta(C))$ time in total on it.
        Overall, the computation of $dp(\nu, \cdot, \ell)$ and $dp'(\nu, \cdot, \ell)$ takes time $\Oh(n_{\nu} + \delta(A) + \delta(B) + \delta(C))$ in total, which is equal to $\Oh(n_{\nu})$ due to \cref{lm:ms-core-bound}.

        Across all $\ell$, the algorithm takes time $\Oh(n_{\nu} \tau)$.
        There are $\Oh(\log n)$ recursion levels, and the values $n_{\nu}$ sum up to $n$ on each level, so the total time complexity is $\Oh(n \tau \log n)$.
    \end{claimproof}

    This concludes the preprocessing phase, which takes $\Oh(n \log^2 n + n \tau \log n) = \Oh(n \log^{3 - \alpha} n)$ time in total.
    We now describe how to answer a query.
    Suppose that we are asked to report the longest increasing subsequence of $s\fragmentco{i_0}{j_0}$ for some $0\le i_0 < j_0 \le n$.
    Let $\ell_0 = \lis(s\fragmentco{i_0}{j_0})$.
    Note that $M^{s}_{i_0, j_0} = \ell_0$ and every longest path from $(i_0, 0)$ to $(j_0, n)$ in $G^s$ corresponds to a longest increasing subsequence of $s\fragmentco{i_0}{j_0}$  due to \cref{lm:ms-meaning}.
    We backtrack one such path through the recursion tree of \cref{lm:lis-dc}.
    Let the current backtracking recursive call be $(\nu, i, j)$ (initially, it is $(\rho, i_0, j_0)$).
    Let $\ell \coloneqq \lis(s_\nu\fragmentco{i}{j})$.
    If $dp(\nu, i, \tau) > j$, we have $\ell < \tau$.
    We then set $\ell_0 = 1$ and increment it as long as $dp(\nu, i, \ell_0 + 1) \le j$, that is, as long as $\ell_0 < \ell$.
    We use $dp'$ to reconstruct the underlying longest increasing subsequence.

    Hence, from now on we assume that $dp(\nu, i, \tau) \le j$, and thus $\ell \ge \tau$.
    Note that $M^{s_\nu}$ was obtained as the $(\max, +)$-product of $M^{s_\nu^{\lo}}$ and $M^{s_\nu^{\hi}}$ for two subpermutations $s_\nu^{\lo}$ and $s_\nu^{\hi}$ that $s_\nu$ was decomposed into.
    Due to \cref{lm:product-witness-reconstruction}, in time $\Oh(\log n)$ we can reconstruct the witness $h$ such that $M^{s_\nu}_{i, j} = M^{s_\nu^{\lo}}_{i, h} + M^{s_\nu^{\hi}}_{h, j}$, that is, some longest increasing subsequence in $s_\nu\fragmentco{i}{j}$ uses the values from $s_\nu^{\lo}$ in $\fragmentco{i}{h}$, and then the values from $s_\nu^{\hi}$ in $\fragmentco{h}{j}$.
    In constant time, we can recompute the indices in $s_\nu^{\lo *}$ corresponding to $i$ and $h$ in $s_\nu^{\lo}$ using the $\rpos$ mapping for $s_\nu^{\lo}$ and recursively compute the prefix of the longest increasing subsequence in $s_\nu^{\lo}\fragmentco{i}{h}$.
    By analogously recursing in the other direction, we compute the suffix of the longest increasing subsequence in $s_\nu^{\hi}\fragmentco{h}{j}$.

    It remains to analyze the time complexity of the reporting procedure.
    Note that if we have $dp(\nu, i, \tau) > j$ in a recursive call, then we first find $\ell$ in $\Oh(\ell)$ time and then reconstruct the underlying subsequence in $\Oh(\ell)$ time.
    Thus, all leaf recursive calls take $\Oh(\ell_0)$ time in total.

    We now prove that there are $\Oh(\ell_0 \log n / \tau)$ internal recursive calls.
    Call a leaf recursive call \emph{heavy} if in it we have $\ell \ge \tau / 2$.
    We claim that each internal recursive call $\nu$ has at least one heavy leaf recursive call in its backtracking subtree.
    Consider the deepest internal recursive call $\mu$ in the subtree of $\nu$.
    As $\mu$ is the deepest internal recursive call, both its children are leaf recursive calls.
    As $\mu$ is not a leaf recursive call, we have that $\ell_{\mu} \ge \tau$ for it.
    As $\ell_{\mu} = \ell_{\mu^L} + \ell_{\mu^R}$, we have that $\max\{\ell_{\mu^L}, \ell_{\mu^R}\} \ge \tau / 2$, and thus either $\mu^L$ or $\mu^R$ is a heavy leaf recursive call.
    Furthermore, note that, as the sum of $\ell$ over all leaf recursive calls is equal to $\ell_0$, there are at most $2 \ell_0 / \tau$ heavy leaf recursive calls.
    Finally, as the depth of the recursion is $\Oh(\log n)$, each heavy leaf recursive call has $\Oh(\log n)$ ancestors, and thus there are $\Oh(\ell_0 \log n / \tau)$ internal recursive calls in total as required.

    Each internal recursive call takes $\Oh(\log n)$ time dominated by the witness reconstruction.
    Therefore, the total time complexity is $\Oh(\ell_0 + \ell_0 \log n / \tau \cdot \log n) = \Oh(\ell_0 \log^{\alpha} n)$, thus proving the claim.
\end{proof}

\subsection{Open Problem: Weighted Range LIS}\label{sec:open}

Since \cref{lm:core-based-multiplication-algorithm} works for arbitrary Monge matrices and not only for unit-Monge matrices, the algorithm of \cref{lm:lis-no-reporting} can be easily extended to \emph{weighted} range LIS queries, where each element of $s$ is equipped with a positive weight and the queries ask to find maximum-weight increasing subsequences of $s[i\dd j)$.
For this, in \cref{def:alignment-graph}, we make the weights of the diagonal edges equal to the weights of the corresponding elements of $s$, and the weights of backward edges larger than all weights in $s$.
The only difference compared to the unweighted version of the problem is its time complexity, which depends on the core size of $M^s$.
In the unweighted case we have $\delta(M^s) = \Oh(n)$ due to \cref{lm:simple-core-bound,lm:calculate-value-from-core} because all elements of $M^s$ are bounded by $\Oh(n)$.
In the weighted case, we are only aware of the trivial bound $\delta(M^s)=\Oh(n^2)$ if the individual weights are $\Omega(n)$.

However, if $\delta(M^s) \le \tOh(n^e)$ holds some constant $e \ge 1$ and every length-$n$ sequence $s$, the whole algorithm takes time $\tOh(n^e)$ for preprocessing and $\Oh(\log n)$ time for queries.
To the best of our knowledge, no algorithm for weighted range LIS with truly sub-quadratic preprocessing time and poly-logarithmic query time is known.
Therefore, a bound of the form $\delta(M^s) \le \tOh(n^e)$ for $e < 2$ valid for every length-$n$ sequence $s$ would be of interest.

One way of proving such a bound could be a strengthened version of \cref{lm:core-preservation}.
Suppose that every two $n\times n$ Monge matrices $A$ and $B$ satisfy $\delta(A \minplus B) \le c \cdot (\delta(A) + \delta(B)) + \tOh(n)$ for some constant $c$; in particular, \cref{lm:core-preservation} proves this fact for $c = 2$.
We can use the divide-and-conquer procedure of \cref{lm:lis-dc} to show a bound on $M^s$ in terms of $c$.
Note that in a single recursive call, we recurse into two subproblems of size $n / 2$, ``desparsify'' the resulting matrices, and then $(\max, +)$-multiply them.
The desparsification step adds a linear number of core elements to the matrices.
Therefore, the core size of $M^s$ can be bounded using the following recursive relation:
\begin{align*}
    f(1) &= \Oh(1),\\
    f(n) &= c \cdot ((f(n / 2) + \Oh(n)) + (f(n / 2) + \Oh(n))) + \tOh(n) = 2c \cdot f(n / 2) + \tOh(n),
\end{align*}
which yields $\delta(M^s) \le f(n) \le \tOh((2c)^{\log_2 n}) = \tOh(n^{1 + \log_2 c})$ for any $c \ge 1$.
For $c = 2$ from \cref{lm:core-preservation}, we recover a trivial bound $\delta(M^s) \le \Oh(n^2)$, but for any $c<2$, this would give a truly sub-quadratic bound on $\delta(M^s)$.
Therefore, any strengthening of \cref{lm:core-preservation} would immediately yield a new algorithm for weighted range LIS queries.

\end{document}